\documentclass[sigplan,10pt,screen]{acmart}
\renewcommand\footnotetextcopyrightpermission[1]{}
\usepackage{xspace}
\usepackage{tabularx}
\usepackage{subcaption}
\usepackage{enumitem}
\usepackage{algorithm}
\usepackage{algorithmicx}
\usepackage[noend]{algpseudocode}
\usepackage{pifont}
\usepackage{multirow}
\usepackage{booktabs}
\usepackage{graphicx}
\usepackage{etoolbox}
\usepackage[all]{nowidow}
\usepackage{colortbl}
\usepackage{latexsym}
\usepackage{svg}
\usepackage{rotating}
\usepackage{amsmath}
\DeclareMathOperator*{\argmin}{arg\,min}

\setitemize{noitemsep,topsep=0pt,parsep=0pt,partopsep=0pt}

\newcommand{\parabf}[1]{\medskip\noindent\textbf{#1}}

\newcommand{\paraf}[1]{\noindent\textbf{#1}}
\newcommand{\cut}[1]{}

\newcommand{\sysname}{RAGCache\xspace}

\begin{document}
\title{\sysname: Efficient Knowledge Caching for Retrieval-Augmented Generation}
\pagestyle{plain}

\author{
Chao Jin$^1$\qquad Zili Zhang$^1$\qquad Xuanlin Jiang$^1$\qquad Fangyue Liu$^1$\\ 
Xin Liu$^2$ \qquad Xuanzhe Liu$^1$\qquad Xin Jin$^1$\\
$^1$\textit{Peking University} \qquad $^2$\textit{ByteDance Inc.}
}

\begin{abstract}
Retrieval-Augmented Generation (RAG) has shown significant improvements
in various natural language processing tasks by integrating the strengths of
large language models (LLMs) and external knowledge databases.
However, RAG introduces long sequence generation and leads to high computation and
memory costs. 
We propose \sysname, a novel multilevel dynamic caching system tailored for
RAG. Our analysis benchmarks current RAG systems,
pinpointing the performance bottleneck (i.e., long sequence due to knowledge
injection) and optimization opportunities (i.e., caching knowledge's intermediate states).
Based on these insights, we design \sysname, which organizes the intermediate states 
of retrieved knowledge in a knowledge tree and caches them in the GPU and host memory hierarchy. 
\sysname proposes a replacement policy that is aware of LLM inference characteristics and 
RAG retrieval patterns. It also dynamically overlaps the retrieval and inference steps to minimize the 
end-to-end latency. We implement \sysname and evaluate it on
vLLM, a state-of-the-art LLM inference system and Faiss, a state-of-the-art vector database.
The experimental results show that \sysname reduces the time to first token (TTFT) by up to 4$\times$ and 
improves the throughput by up to 2.1$\times$ compared to vLLM integrated with Faiss.

\end{abstract}

\maketitle

\section{Introduction}
\label{sec:introduction}

Recent advancements in large language models (LLMs) like GPT-4~\cite{gpt4}, LLaMA2~\cite{touvron2023llama}, 
and PalM~\cite{palm} have significantly enhanced performance across various natural language processing (NLP) tasks,
including question answering, summarization, and translation~\cite{siriwardhana2023improving, zhang2024benchmarking, zhang2023prompting}. 
Retrieval-augmented generation (RAG)~\cite{lewis2020retrieval, langchain} further enhances LLMs by incorporating 
contextually relevant knowledge from external databases, such as Wikipedia~\cite{wikipedia_embeddings}, to 
improve the generation quality. With informative external knowledge, RAG have achieved comparable or even 
better performance than LLMs fine-tuned for specific downstream tasks~\cite{chen2024benchmarking}.

For an RAG request, the RAG system first retrieves relevant documents
from the knowledge database. The documents
are typically represented as feature vectors in a vector database through embedding models,
and the retrieval step is implemented by vector similarity search. 
Then, RAG injects the retrieved documents (i.e., external knowledge)
into the original request and feeds the augmented request to the LLM for generation.
With the help of the retrieved documents, RAG expands LLMs' knowledge base and
contextual understanding, thereby improving the generation quality~\cite{chen2024benchmarking}.

With knowledge injection, RAG introduces long sequence generation for the augmented
request, which leads to high computation and memory costs. For instance, the initial
request contains 100 tokens, and the retrieved documents may contain 1000 tokens in total.
Consequently, the extra computation and memory costs for the augmented request are $>$10$\times$
higher than the original request.
This escalation in resource requirements poses a substantial challenge
in scaling systems for efficient processing of RAG requests.

Recent work~\cite{sglang, vllm}, focusing on system optimizations of LLM inference, has made significant progress 
in sharing the intermediate states of LLM inference to reduce recomputation costs.
vLLM~\cite{vllm} manages the intermediate states in non-contiguous memory blocks 
to allow fine-grained memory allocation and state sharing for a single request's multiple generation iterations. 
SGLang~\cite{sglang} identifies the reusable intermediate states across different requests for 
LLM applications like multi-turn conversations and tree-of-thought~\cite{yao2024tree}. 
However, these efforts only optimize for LLM inference without considering the characteristics of RAG.
They cache the intermediate states in GPU memory, which has limited capacity considering the long sequences 
in augmented requests, leading to suboptimal performance. 

We conduct a system characterization of RAG, which measures the performance of current RAG systems under 
various datasets and retrieval settings with representative LLMs. 
Our analysis highlights a significant performance limitation rooted in the processing of augmented sequences 
due to document injection. 
In addition, we uncover two potential \emph{opportunities} for system optimizations to mitigate this constraint. 
First, the recurrence of identical documents across multiple requests enables the sharing of LLM inference's 
intermediate states for such documents.
Second, a small fraction of documents accounts for the majority of retrieval requests.
This allows us to cache the intermediate states of these frequently accessed documents to reduce the 
computational burden.

To this end, we propose \sysname, a novel multilevel dynamic caching system tailored for
RAG. \sysname is the first system to cache the intermediate states of
retrieved documents (i.e., external knowledge) and share them across multiple requests. 
The core of \sysname is a knowledge tree that adapts the intermediate states of the retrieved documents to the 
GPU and host memory hierarchy. Documents with more frequent accesses are cached in the fast GPU memory and those with 
fewer accesses are cached in the slow host memory. There are mainly two challenges in designing the caching system for RAG. 

First, RAG systems are sensitive to the referred order of the retrieved documents. 
For instance, there are two documents $D_1$ and $D_2$ and two requests $Q_1$ and $Q_2$.
Let $Q_1$'s and $Q_2$'s relevant documents be $[D_1,D_2]$ and $[D_2,D_1]$, respectively, 
where $[D_1,D_2]$ means $D_1$ is more relevant than $D_2$.
The intermediate states (i.e., the key value tensors) of $[D_1,D_2]$ are different from that of $[D_2,D_1]$ 
because the key-value tensor of the new token is calculated based on the preceding tokens in the
attention mechanism of LLMs~\cite{vaswani2017attention}. 
Unfortunately, we also cannot swap the order of $D_1$ and $D_2$.
As recent efforts have shown, the generation quality of the LLMs will be affected by the referred 
order~\cite{chen2023understanding, liu2024lost}. We use knowledge tree to organize the intermediate states 
of the retrieved documents in the GPU and host memory hierarchy and design a prefix-aware 
Greedy-Dual-Size-Frequency (PGDSF) replacement policy that 
comprehensively considers the document order, size, frequency, and recency to minimize the miss rate. 
We also propose a cache-aware request scheduling approach to further improve the hit rate when the request 
rate is high.

Second, vector retrieval (processed on CPU) and LLM inference (processed on GPU) are two independent steps in RAG.
The two steps are executed sequentially in the current RAG systems, 
leading to idle GPU resources during retrieval and long end-to-end latency.
We propose a dynamic speculative pipelining strategy to dynamically
overlap the computation of the two steps and minimize the end-to-end latency while keeping the system load 
under control.

We implement a \sysname prototype and evaluate it on various datasets and representative LLMs.
The experimental results show that \sysname outperforms the state-of-the-art solution, vLLM~\cite{vllm} 
integrated with Faiss~\cite{pinecone-faiss}, by up to 4$\times$ on time to first token (TTFT) and improves the throughput 
by up to 2.1$\times$. 
Compared to SGLang~\cite{sglang}, which reuse the intermediate states in GPU memory, 
\sysname reduces the TTFT by up to 3.5$\times$ and improves the throughput by up to 1.8$\times$. 

In summary, we make the following contributions.
\begin{itemize}[leftmargin=*]
    \item We conduct a detailed system characterization of RAG, which reveals the performance
    bottleneck and optimization opportunities.

    \item We propose \sysname, to the best of our knowledge, the first RAG system
    that caches the intermediate states of external knowledge and shares them across
    multiple queries to reduce the redundant computation.

    \item We design a prefix-aware GDSF replacement policy that leverages the characteristics of RAG to minimize 
    the miss rate and a dynamic speculative pipelining approach to minimize the end-to-end latency. 

    \item We implement a \sysname prototype. The evaluation shows that \sysname
    outperforms vLLM integrated with Faiss by up to 4$\times$ on TTFT and 2.1$\times$ on throughput. 
    Compared to the state-of-the-art caching system for LLM applications, SGLang, 
    \sysname achieves up to 3.5$\times$ lower TTFT and up to 1.8$\times$ higher throughput.
\end{itemize}
\section{Background}
\label{sec:background}
Retrieval-Augmented Generation (RAG) represents a significant advancement
in the field of natural language processing (NLP) and machine learning,
combining LLMs with the vast information accessible in external knowledge databases. 
Specifically, RAG is employed to enhance the generative models' ability to produce more accurate, relevant, and
contextually rich responses by dynamically retrieving information from a corpus
during the generation process. This hybrid approach combines the strengths of two
major strands: the deep contextual understanding of LLMs
and the precision of knowledge database retrieval.
Recent work~\cite{lewis2020retrieval, jiang2024piperag, borgeaud2022improving, ram2023context, trivedi2022interleaving, langchain} 
has demonstrated that RAG can significantly improve the generation quality
across various benchmarks compared to solely generative models.
The RAG framework has since been applied across various tasks, 
including question answering~\cite{siriwardhana2023improving}, 
content creation~\cite{khattab2022demonstrate}, and even code generation~\cite{lu2022reacc, vaithilingam2022expectation},
showcasing its versatility and promise.

\begin{figure}[t]
    \centering
    \includegraphics[width=1\linewidth]{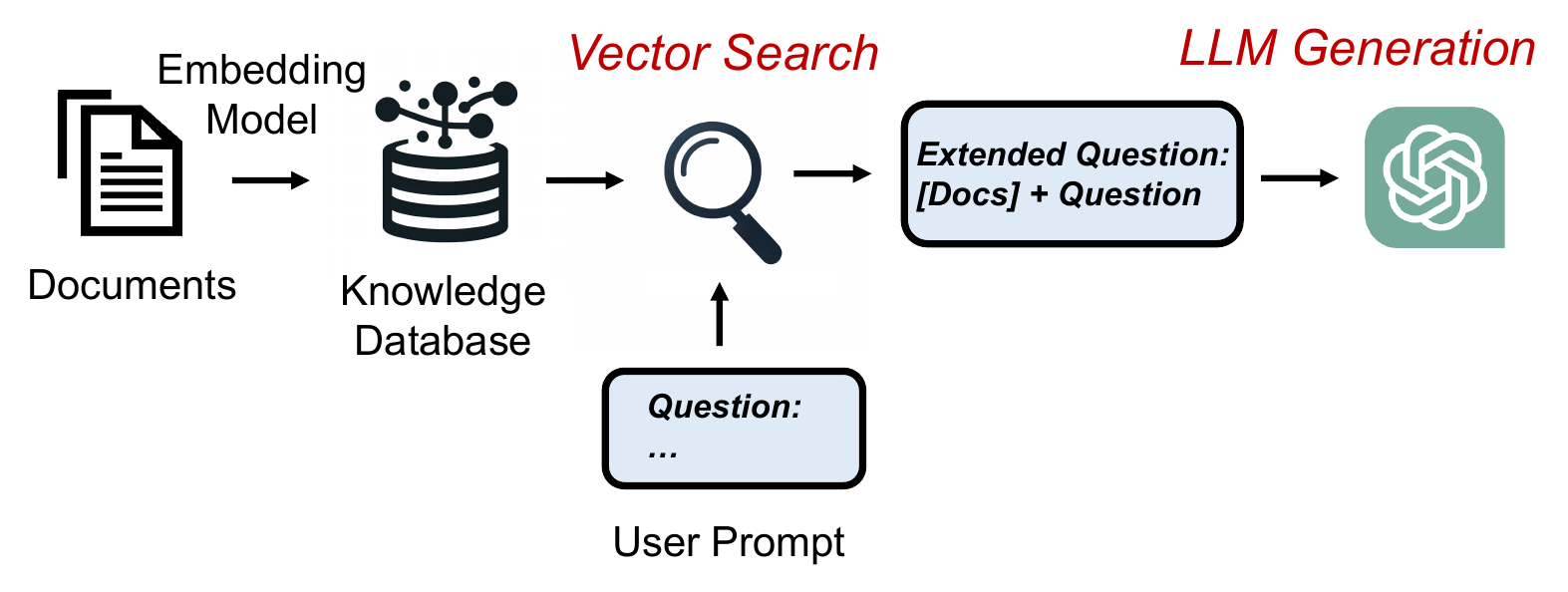}
    \vspace{-0.2in}
    \caption{RAG workflow.}
    \vspace{-0.2in}
    \label{fig:bg:rag}
\end{figure}

As shown in Figure~\ref{fig:bg:rag}, RAG operates on a two-step workflow: \emph{retrieval}
and \emph{generation}, integrating offline preparation with real-time processing for enhanced performance.
Initially, in its offline phase, RAG transforms the external knowledge sources, such as documents,
into high-dimensional vectors using advanced embedding models. RAG then
indexes these vectors into a specialized vector database designed for efficient retrieval.
Upon receiving a user request, RAG first accesses this vector database to conduct a vector similarity search,
retrieving the documents that best match the request based on their semantic content.
Following this, RAG combines the content of these retrieved documents with the original user request,
creating an augmented request. This augmented request is then provided to an LLM,
which leverages the combined information to generate a response that is more informed and contextually rich.

In an RAG workflow, the retrieval step is mainly performed on CPUs, while the generation step is executed on GPUs.
From a system perspective, the end-to-end performance of RAG is affected by both the retrieval or generation steps.
The retrieval time is mainly determined by the vector database's scale,
and the generation time is decided by the model size and the sequence length.
Our subsequent characterization will identify RAG's performance bottleneck and highlight potential areas for optimization.

\section{RAG System Characterization}
\label{sec:characterization}
In this section, we conduct a comprehensive system characterization. First, we
identify the performance bottlenecks for RAG, i.e., the LLM generation step.
Next, we evaluate the performance improvement of caching the intermediate states of the retrieved knowledge.
Finally, we analyze the question patterns and the potential for caching optimization.

\begin{figure}[t]
    \centering
    \includegraphics[width=0.7\linewidth]{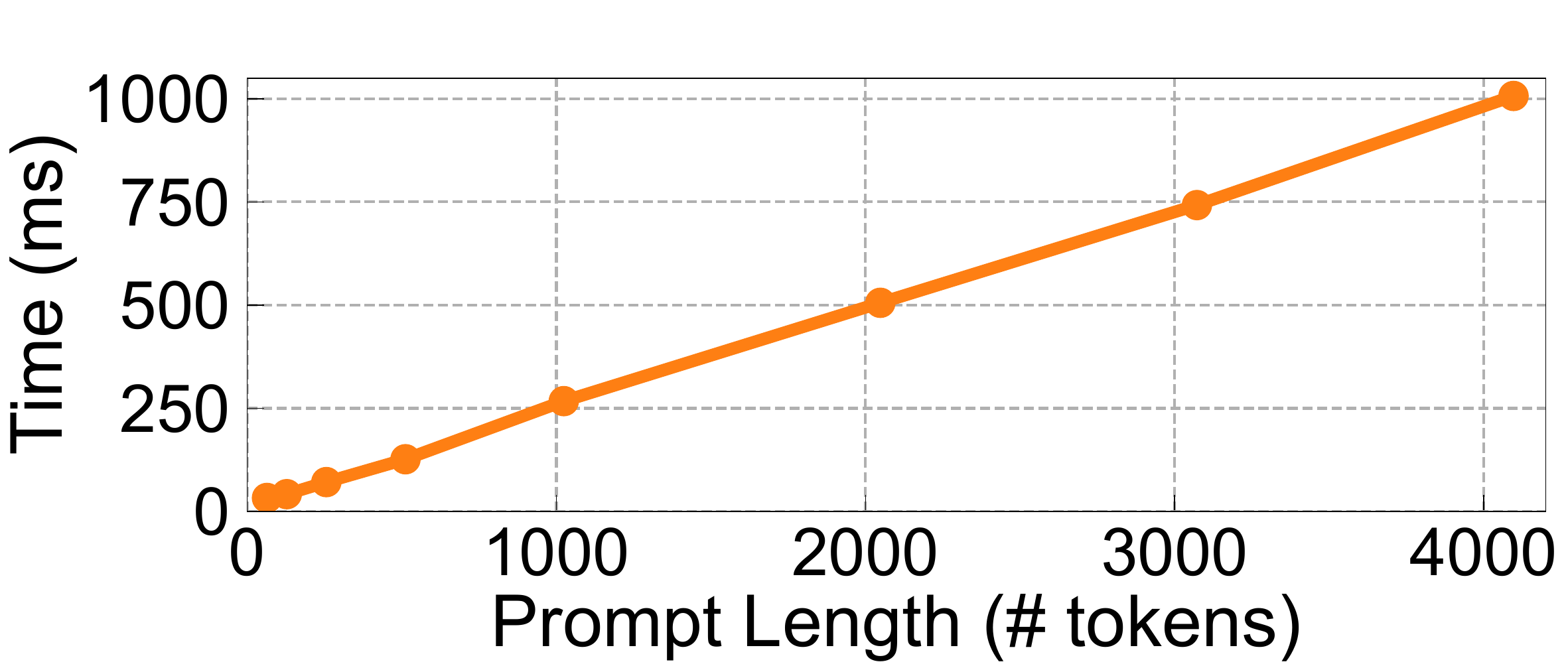}
    \vspace{-0.15in}
    \caption{Inference time with different input lengths.}
    \vspace{-0.15in}
    \label{fig:characterization:inference}
\end{figure}

\begin{figure}[t]
    \centering
    \includegraphics[width=0.7\linewidth]{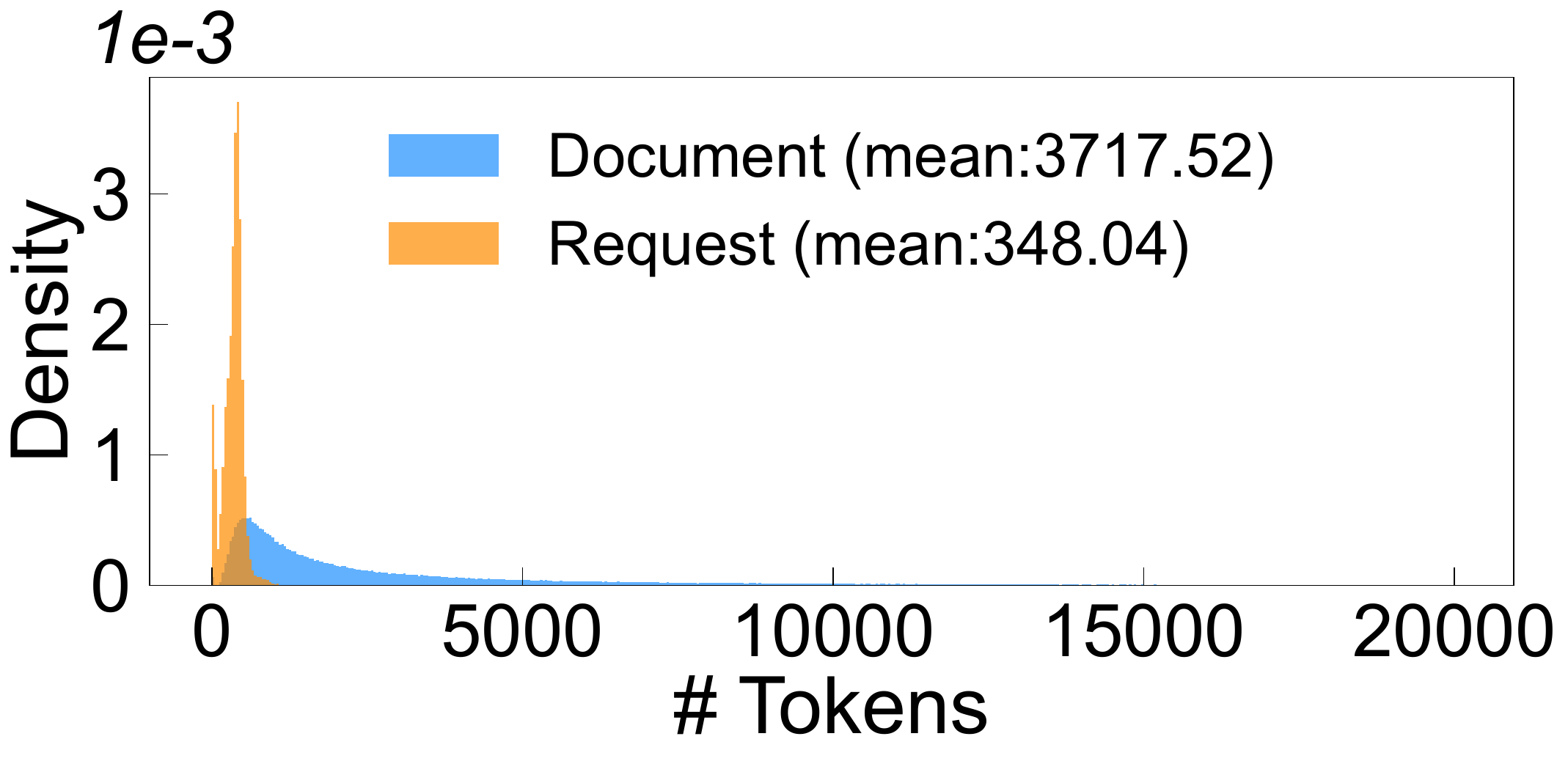}
    \vspace{-0.15in}
    \caption{The distribution of token number.}
    \vspace{-0.15in}
    \label{fig:characterization:distribution}
\end{figure}

\subsection{Performance Bottleneck}
\label{sec:characterization:bottleneck}
LLM inference can be divided into two distinct phases: prefill and decoding.
The prefill phase involves computing the key-value tensors of the input tokens,
while the decoding phase generates the output token in an auto-regressive manner based on the previously generated key-value tensors.
The prefill phase is particularly time-consuming, as it requires to compute the entire input sequence's key-value tensors.

As discussed in \S~\ref{sec:background}, RAG comprises two steps: retrieval and generation.
Recent work~\cite{zhang2023fast, zhang2024fast} shows that 
the retrieval step executes in milliseconds per request with a high accuracy for billion-scale vector databases.
Meanwhile, the generation step, conducted on GPUs, is heavily affected by sequence length and model size.
To identify the performance bottleneck, we evaluate the inference time with fixed output length and different input lengths
on LLaMA2-7B, the smallest model in the LLaMA2 series~\cite{touvron2023llama}.
Larger models will have longer inference time.
The backend system is vLLM~\cite{vllm} equipped with one NVIDIA A10G GPU.
Figure~\ref{fig:characterization:inference} shows that 
the inference time, mainly dominated by the prefill phase,
increases rapidly with sequence length and reaches one second when the sequence length is larger than 4000 tokens.

The sequence length in the LLM generation step is the token number of the original request plus the retrieved document.
We generate a document dataset based on the Wikipedia corpus~\cite{wikipedia_embeddings} with $\sim$0.3 million 
documents from most popular Wikipedia pages. Figure~\ref{fig:characterization:distribution} demonstrates
the distribution of the document length and the request length.
The document length is significantly longer than the request length of the MMLU dataset~\cite{hendrycks2020measuring}.
With an average document length of 3718 tokens, the corresponding inference time is 
markedly higher than the retrieval step in most cases. 
Notably, the retrieval step may take comparable time to the generation step when users require relevant knowledge 
with extremely high accuracy~\cite{zhang2024fast}, which necessitates extensive searching in the vector database 
and further complicates the performance bottleneck.

\subsection{Optimizations Opportunities}
\label{sec:characterization:opp}

\paraf{Caching knowledge.}
The generation step's performance bottleneck primarily arises from processing the
long sequence's key-value tensors in attention blocks. 
A simple yet effective optimization for RAG involves caching these key-value tensors of previously retrieved 
documents. For example, let requests, $Q_1$ and $Q_2$, both refer to the same document, $D_1$.
If $Q_1$ arrives first, the key-value tensors of $D_1$ are computed, and we can cache
the key-value tensors. When $Q_2$ arrives, we can reuse the cached key-value tensors
to reduce the prefill latency of $Q_2$.
The average prefill latency with caching is calculated as follows:
\begin{align}
    \nonumber
    & \mathit{Prefill\ Latency} = \mathit{Miss\ Rate} \times \mathit{Full\ Prefill\ Latency} \\
    \nonumber
    & + \mathit{(1-Miss\ Rate)} \times \mathit{Cache\ Hit\ Latency}
\end{align}
To explore caching's optimization opportunity, we consider three crucial factors:
full prefill computation, cache hit, and miss rate.

\begin{figure}[t]
    \centering
    \includegraphics[width=0.7\linewidth]{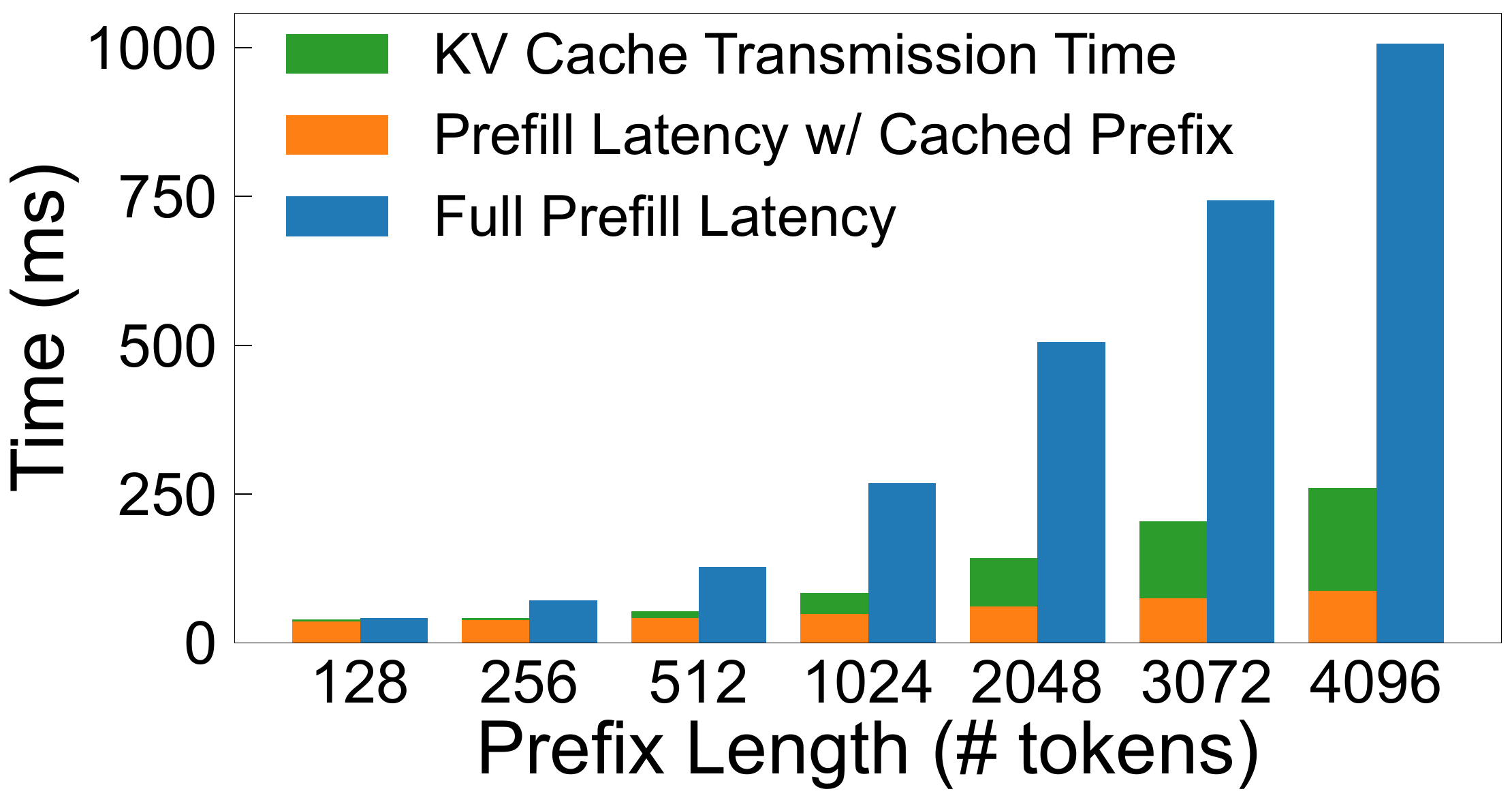}
    \vspace{-0.15in}
    \caption{Prefill latency characterization.}
    \vspace{-0.15in}
    \label{fig:characterization:time_comparison}
\end{figure}

\parabf{Full prefill computation.}
To quantify the full computation, we compare the LLM prefill phase's latency with and without
caching these partial intermediate states (i.e., the key-value cache of prefixes).
We set the original request length to 32 tokens and vary the prefix length from 128 to 4096 tokens. 
Figure~\ref{fig:characterization:time_comparison}
illustrates that the prefill latency is significantly
reduced when caching is employed. 
In the scenario of cached prefix, only the request tokens' key-value tensors are computed.
Conversely, in the scenario of full prefill computation, the key-value tensors of the entire sequence need to 
be calculated.
The full prefill latency is up to 11.5$\times$ lower than that in the cached prefix scenario.
These results underscore the substantial performance
improvement achieved by caching intermediate states of accessed documents.

\begin{figure*}[t]
    \centering
    \includegraphics[width=1\linewidth]{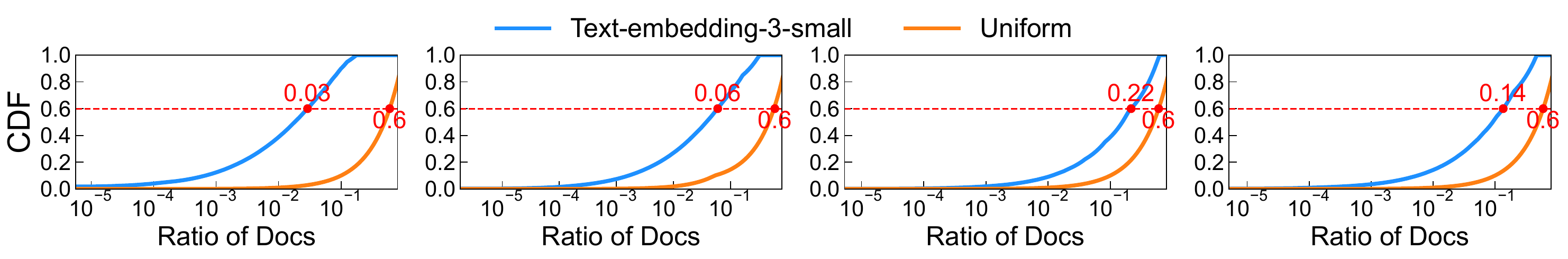}
    \newline
    \hspace*{0.5em}
        \begin{subfigure}{0.24\linewidth}
            \caption{MMLU.}
            \label{fig:characterization:parttern:mmlu}
        \end{subfigure}
        \begin{subfigure}{0.24\linewidth}
            \caption{Google Natural Questions.}
            \label{fig:characterization:parttern:googlenq}
        \end{subfigure}
        \begin{subfigure}{0.24\linewidth}
            \caption{HotpotQA.}
            \label{fig:characterization:parttern:hotpotqa}
        \end{subfigure}
        \begin{subfigure}{0.24\linewidth}
            \caption{TriviaQA.}
            \label{fig:characterization:parttern:triviaqa}
        \end{subfigure}
    \vspace{-0.1in}
    \caption{Retrieval pattern on different datasets.}
    \vspace{-0.15in}
    \label{fig:characterization:parttern}
\end{figure*}

\parabf{Cache hit.}
The cache hit comprises two components: prefill computation of the request tokens and
loading the key-value cache of the retrieved documents.
The former is negligible compared to miss penalty. 
As for the latter one, the limited GPU memory contrasts sharply with the substantial size
of the key-value cache from retrieved documents.
This discrepancy necessitates leveraging the host memory to extend
the caching system, accommodating a greater volume of documents.
However, this introduces a potential overhead: the transmission
of key-value cache between the GPU and host memory.
To assess this, we conduct an evaluation of the transmission
overhead. Figure~\ref{fig:characterization:time_comparison} adds the KV cache transmission time with the 
given prefix length to the prefill time with cached prefix, representing the cache hit latency.
Even with the transmission overhead, the cache hit latency
is significantly better (up to 3.9$\times$ lower) than the full prefill latency, highlighting the advantages of caching
intermediate states of retrieved documents. 

\begin{figure}[t!]
    \centering
    \begin{subfigure}{0.48\linewidth}
        \includegraphics[width=0.95\linewidth]{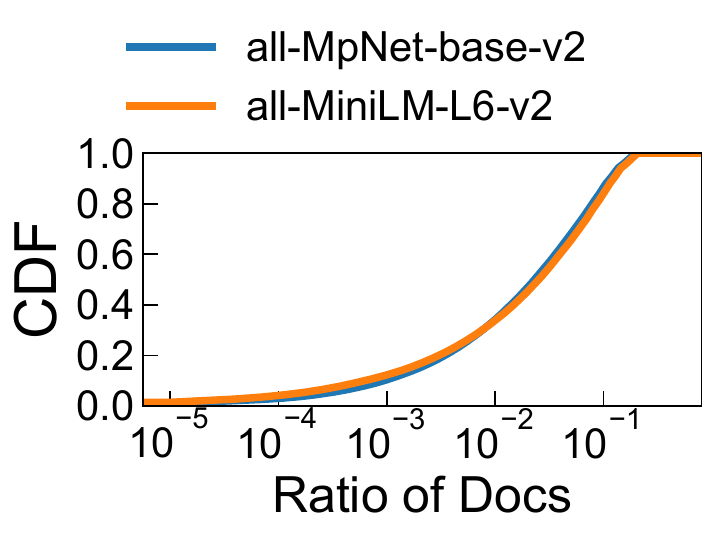}
        \caption{Different embedding models.}
        \label{fig:characterization:query_pattern:embedding}
    \end{subfigure}
    \begin{subfigure}{0.48\linewidth}
        \includegraphics[width=0.95\linewidth]{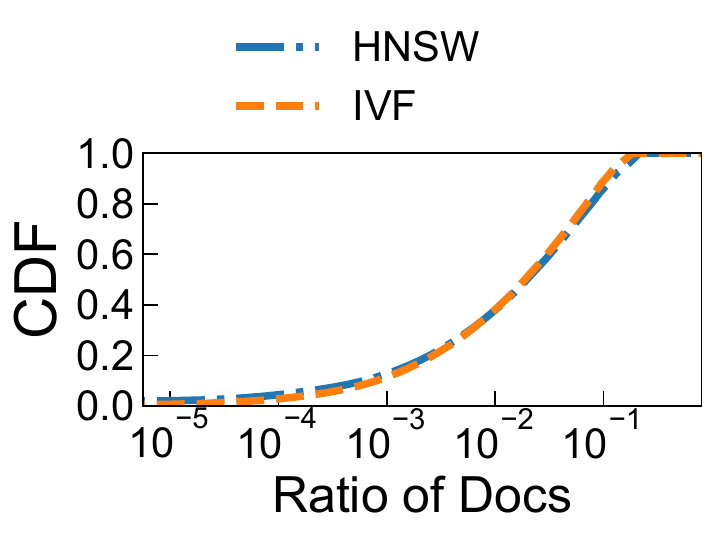}
        \caption{Different ANN indexes.}
        \label{fig:characterization:query_pattern:index}
    \end{subfigure}
    \vspace{-0.1in}
    \caption{Retrieval pattern under different settings.}
    \vspace{-0.15in}
    \label{fig:characterization:query_pattern:settings}
\end{figure}

\parabf{Miss rate.}
The final consideration lies in the retrieval pattern of RAG systems. The cache performance
is dominated by the miss rate, which is directly influenced by the retrieval pattern.
For example, a 100\% cache miss rate occurs when each request retrieves a unique document.
In such a scenario, caching intermediate states of retrieved documents is meaningless.
We analyze the document retrieval pattern in four representative question-answering
datasets for RAG: MMLU~\cite{hendrycks2020measuring}, Google Natural Questions~\cite{kwiatkowski2019natural}, 
HotpotQA~\cite{yang2018hotpotqa}, and TriviaQA~\cite{joshi2017triviaqa}.
We convert the documents on Wikipedia to vectors through \texttt{text-embedding-3-small} model~\cite{text-embedding-3} 
from OpenAI~\cite{openai} for retrieval. 
The number of documents referred by one request is top-$1$. The ANN index is FlatL2, i.e.,
exact search on the entire dataset with Euclidean distance.
Figure~\ref{fig:characterization:parttern} shows the CDF of the accessed documents.
We observe that the retrieval pattern is skewed, with a small fraction of documents
accounting for the majority of retrieval requests. For example, the top 3\% documents
are referred to by 60\% requests in the MMLU dataset,
which is 20$\times$ less than the uniform distribution.
This observation reveals a low miss rate to cache the frequently accessed documents.

Further analysis on additional embedding models and ANN indexes for vector search is shown in 
Figure~\ref{fig:characterization:query_pattern:settings}.
All of the results exhibit a similar retrieval pattern no matter which embedding model or ANN index is used.
The results are consistent with FlatL2 index, which indicates the potential for caching optimization under different settings.

\section{\sysname Overview}
We present \sysname, a novel multilevel dynamic caching system tailored for
RAG. \sysname caches the key-value tensors of retrieved documents across multiple
requests to minimize redundant computation. 
The core of \sysname is a knowledge tree with a prefix-aware 
Greedy Dual-Size Frequency (PGDSF) replacement policy that ensures caching the most critical key-value tensors.
\sysname also implements a global RAG controller that orchestrates interactions between the external knowledge database and LLM inference engine.
The controller is enhanced by system optimizations including cache-aware reordering and dynamic speculative pipelining.

\parabf{Architecture overview.}
We provide a brief overview of \sysname in Figure~\ref{fig:overview:arch}.
When a request arrives, the RAG controller first retrieves the relevant documents from the external knowledge database.
These documents are then forwarded to the cache retriever to locate matching key-value tensors.
If the key-value tensors are absent from the cache, \sysname directs the LLM inference engine to produce new tokens.
Conversely, if the tensors are available, the request with the key-value tensors are forwarded
to the LLM inference engine, which then employs a prefix caching kernel for token generation.
After generating the first token, the key-value tensors are relayed back to the RAG controller, 
which caches the tensors from the accessed documents and refreshes the cache's status.
Finally, the generated answer is delivered to the user as the response.

\parabf{Cache retriever.}
The cache retriever efficiently locates the key-value tensors for documents
stored in the in-memory cache, utilizing a knowledge tree to organize these tensors.
This tree, structured as a prefix tree based on document IDs, aligns with the LLM's 
position sensitivity to the document order. Each path within this tree represents one
specific sequence of documents referenced by a request, with each node holding the key-value tensor
of a referred document.
Different paths may share the same nodes, which indicates the shared documents across different requests.
This structure enables the retriever to swiftly access the key-value tensors of documents in their specified order.

\parabf{RAG controller.}
The RAG controller orchestrates the interactions with some system optimizations tailored for RAG.
Prefix-aware Greedy Dual-Size Frequency (PGDSF) policy is employed to minimize the cache miss rate.
PGDSF calculates a priority based on the frequency, size of key-value tensors, last access time, and 
prefix-aware recomputation cost. 
The cache eviction is determined by the priority, which ensures the most valuable tensors are retained.
Cache-aware reordering schedules the requests to improve cache hit rate and prevent thrashing,
while also ensuring request fairness to mitigate starvation issues.
Dynamic speculative pipelining is designed to overlap the knowledge retrieval and LLM inference to minimize 
the latency. This optimization leverages the mid-process generation of retrieval results to initiate LLM inference 
early. 

\begin{figure}[t]
    \centering
    \includegraphics[width=0.7\linewidth]{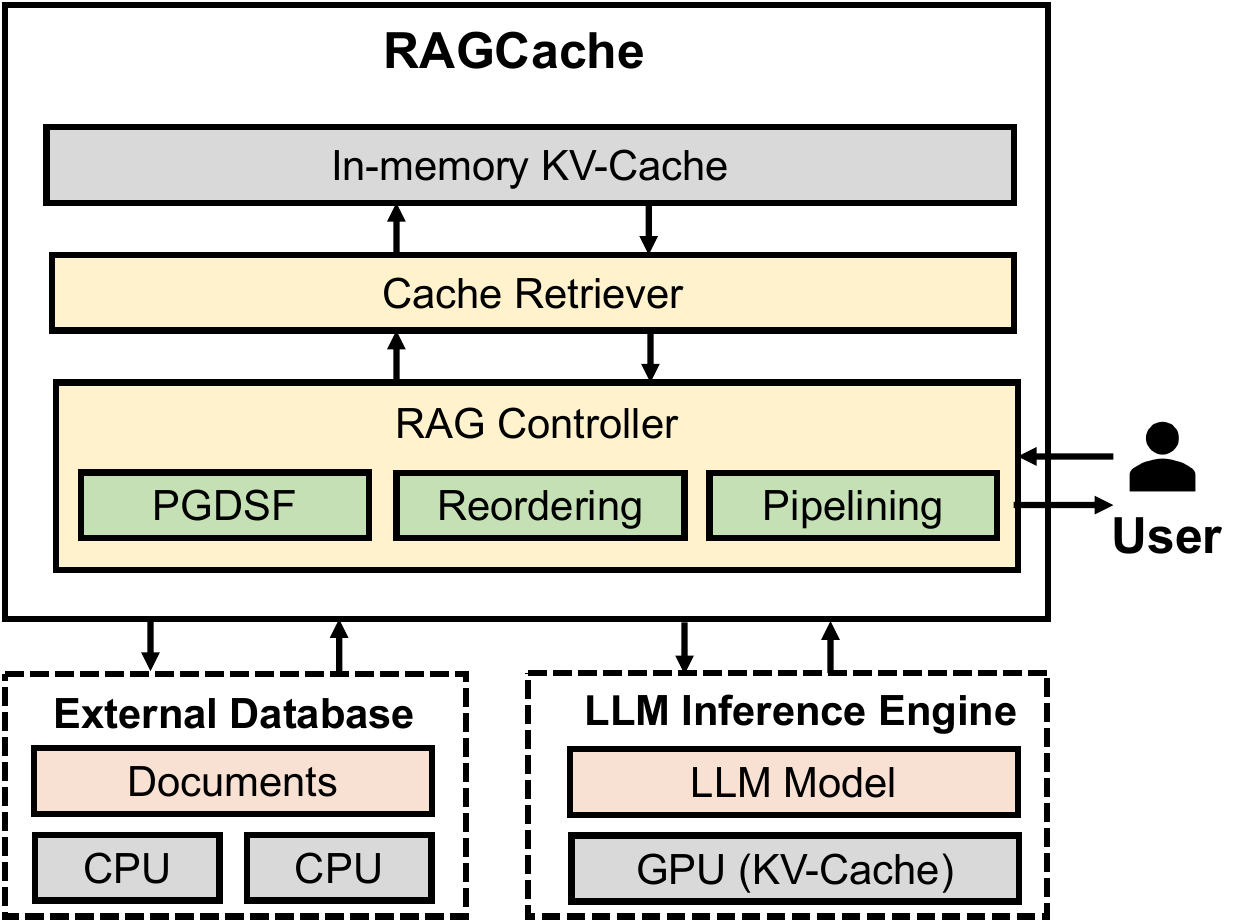}
    \vspace{-0.1in}
    \caption{\sysname overview.}
    \vspace{-0.15in}
    \label{fig:overview:arch}
\end{figure}
\section{\sysname Design}
\label{sec:design}
In this section, we present the design of \sysname. We first introduce the cache structure and the
prefix-aware replacement policy (\S\ref{sec:design:cache}). Then, we describe the 
cache-aware reordering strategy to improve the cache hit rate (\S\ref{sec:design:reorder}).
Finally, we present the dynamic speculative pipelining approach to overlap
knowledge retrieval and LLM inference (\S\ref{sec:design:speculative}).

\subsection{Cache Structure and Replacement Policy}
\label{sec:design:cache}
Different from traditional cache systems that cache individual objects, \sysname caches
the key-value tensors of the retrieved documents that are sensitive to the referred order.
For example, consider two document sequences: $[D_1, D_3]$ with key-value tensors $KV$ and $[D_2, D_3]$ with $KV'$.
Although $KV[1]$ and $KV'[1]$ both pertain to $D_3$,
they are different in values. This discrepancy arises because the key-value tensor
for a given token is generated based on the preceding tokens,
underscoring the order-dependence of key-value tensors.

\begin{figure}[t]
    \centering
    \includegraphics[width=0.8\linewidth]{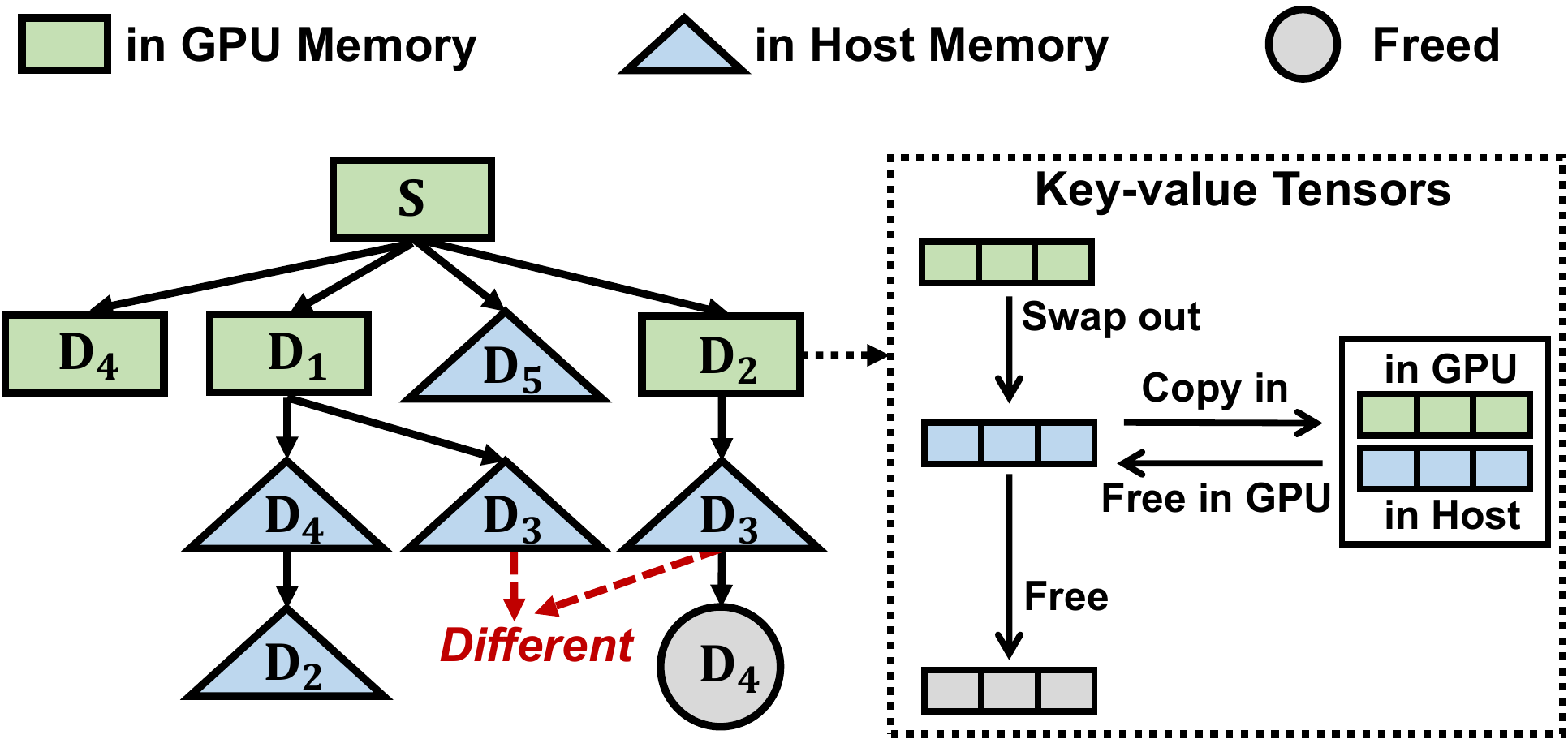}
    \vspace{-0.1in}
    \caption{Knowledge tree.}
    \vspace{-0.2in}
    \label{fig:design:tree}
\end{figure}

To facilitate fast retrieval while maintaining the document order, \sysname structures the documents' key-value tensors
with a knowledge tree, as depicted in Figure~\ref{fig:design:tree}. 
This tree assigns each document to a node, which refers to the memory addresses of the
document's key-value tensors. Following vLLM~\cite{vllm}, \sysname stores 
the key-value tensors in non-continuous memory blocks for KV cache reuse.
The root node $S$ denotes the shared system prompt.
A path from the root to a particular node represents a sequence of documents.

This design inherently allows \sysname to serve multiple requests simultaneously through overlapping paths in the 
tree. \sysname retrieves tensors by prefix matching along these paths. During the prefix matching process, 
if a subsequent document is not located among the child nodes, the traversal is promptly terminated, 
and the identified document sequence is returned. 
This method ensures efficiency with a time complexity of $O(h)$, where $h$ represents the tree's height. 

\parabf{Prefix-aware Greedy-Dual-Size-Frequency (PGDSF) replacement policy.} 
With the knowledge tree, \sysname has to decide each node's placement within a hierarchical cache.
Nodes that are accessed more frequently are ideally stored in GPU memory for faster access speeds,
while those accessed less often are allocated to the slower host memory or simply freed.
To optimize node placement, \sysname employs a prefix-aware Greedy-Dual-Size-Frequency (PGDSF) replacement policy,
which is based on the classic GDSF policy~\cite{cherkasova1998improving}.
Unlike traditional caching strategies such as LRU, which neglect the variable sizes of documents,
PGDSF evaluates each node based on its access frequency, size, and access cost.
This method utilizes limited storage capacity by maintaining the most beneficial nodes, whose \emph{priority} is 
defined as follows:
\begin{align}
    \label{formula:GDSF1}
    & \mathit{Priority} = \mathit{Clock} + \frac{\mathit{Frequency} \times \mathit{Cost}}{\mathit{Size}}
\end{align}
Nodes with lower priority are evicted first.
$\mathit{Clock}$ tracks node access recency.
We maintain two separate logical clocks in the RAG controller for GPU and host memory, respectively, 
to adapt to the cache hierarchy. Each clock starts at zero and updates with every eviction. When 
a document is retrieved, its node's clock is set and its priority is adjusted.
Nodes with older clock, indicating less recent use, receive lower priorities.
Let $E$ be the set of evicted nodes in one eviction operation. The clock is updated accordingly:
\begin{align}
    \label{formula:GDSF2}
    & \mathit{Clock} = \max_{n \in E} \mathit{Priority}(n)
\end{align}
$\mathit{Frequency}$ represents the total retrieval count for a document within a time window. 
This count is reset to zero upon system start or cache clearance. The priority is proportional to the frequency, 
and thus more frequently accessed documents have higher priorities.
$\mathit{Size}$ reflects the number of tokens in a document post-tokenization, directly influencing the memory 
required for its key-value tensors.
$\mathit{Cost}$, defined as the time taken to compute a document's key-value tensors, varies with GPU 
computational capacity, document size, and the sequence of preceding documents.

\begin{figure}[t]
    \centering
    \includegraphics[width=0.7\linewidth]{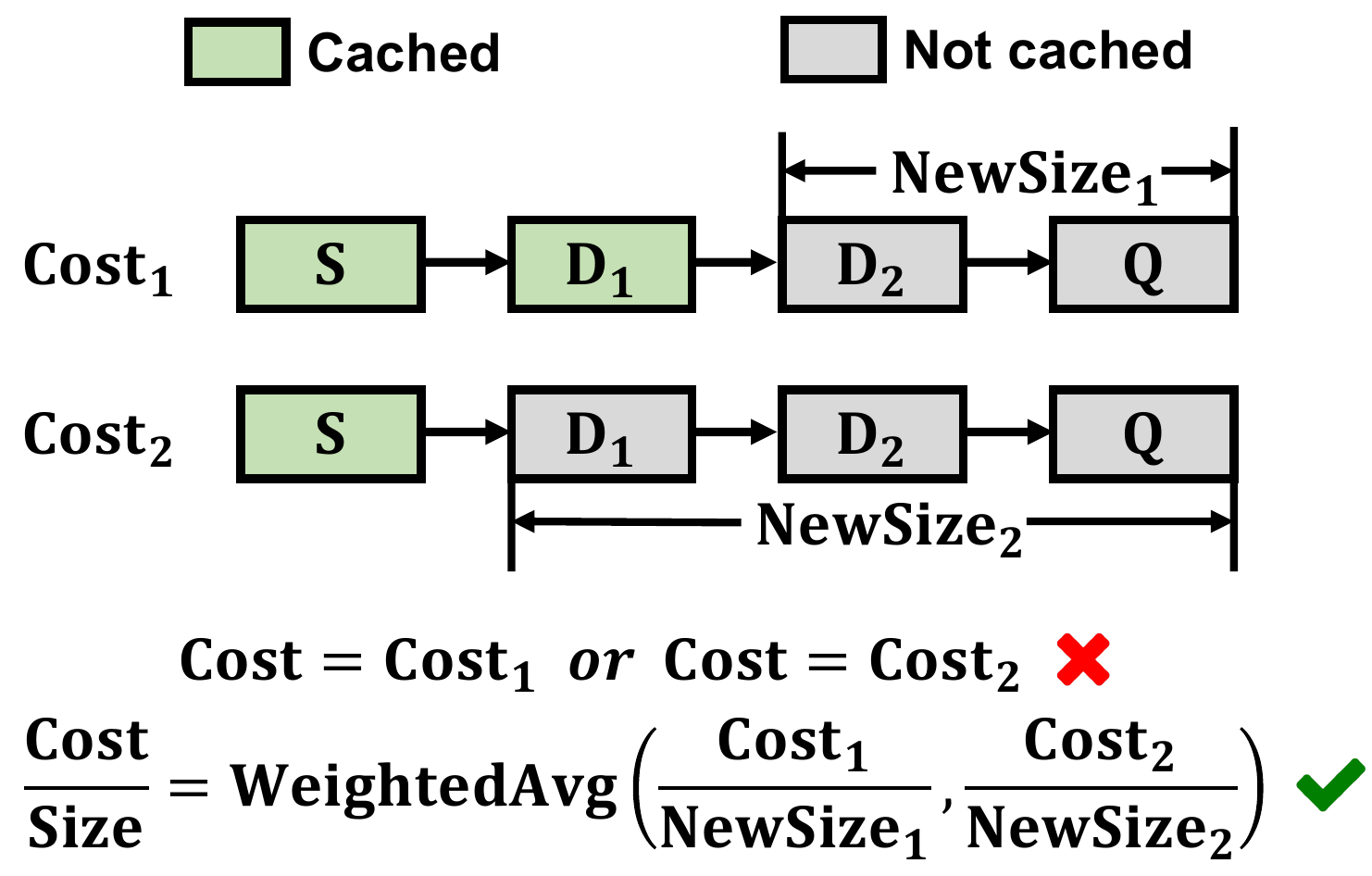}
    \vspace{-0.1in}
    \caption{Cost estimation in PGDSF.}
    \vspace{-0.2in}
    \label{fig:design:cost}
\end{figure}

PGDSF achieves prefix awareness for RAG systems in two aspects: $Cost$ estimation and node placement. 
Unlike GDSF, where costs are straightforward (e.g., object size in web caching), 
RAG costs involve complex LLM generation dynamics. For example, Figure~\ref{fig:design:cost} shows 
varying costs incurred by the same request denoted as 
$[S, D_1, D_2, Q]$. To estimate the cost for $D_2$, directly using the the cost where 
$[S, D_1]$ is cached or only $S$ is cached is imprecise. Besides, the latter case's cost also includes 
the time to compute the key-value tensors for $D_1$ and $Q$.
PGDSF addresses this problem by replacing $Cost/Size$ in Formula~\ref{formula:GDSF1} as follows:
\begin{align}
    \label{formula:Cost}
    & \frac{\mathit{Cost}}{\mathit{Size}} = \frac{1}{m}\sum_{i=1}^{m} \frac{\mathit{Cost}_i}{\mathit{NewSize_i}}
\end{align}
where $m$ is the number of requests that access the document but do not have the document cached.
$\mathit{Cost}_i / \mathit{NewSize_i}$ represents the compute time per non-cached token for the $i$-th request.
Such estimation inherently considers the document size by amortizing the cost to all non-cached tokens.
As for $\mathit{Cost_i}$, \sysname profiles the LLM prefill time with varying cached and non-cached token lengths 
offline and uses bilinear interpolation to estimate the cost for a given request. 
Document retrieval triggers an update in node frequency, cost estimation and clock within the knowledge tree, 
or initiates a new node for documents not previously cached. 

PGDSF orchestrates node placement in the knowledge tree, which is divided into GPU, host, and free segments, 
as illustrated in Figure~\ref{fig:design:tree}.
Nodes in GPU memory serve as parent nodes to those in host memory, 
establishing a hierarchical structure.
\sysname dynamically manages node eviction across these segments for efficiency. 
Specifically, when the GPU memory is full, \sysname swaps
the least priority node in leaf nodes to the host memory. 
\sysname applies a similar process for host memory oversubscription.
This eviction strategy upholds the tree's hierarchical partitioning, which is pivotal for aligning with 
the memory hierarchy and prefix sensitivity in LLM generation.
A node relies on its parent node for key-value tensor calculation, emphasizing the need for prioritizing parent 
node placement for rapid retrieval.

\begin{algorithm}[t]
    \caption{Knowledge Tree Operations}
    \label{alg:design:tree}
    \begin{footnotesize}
    \begin{algorithmic}[1]
        \Function{UPDATE\_NODE\_IN\_GPU}{$node$, $is\_cached$, $\alpha$, $\beta$}
            \State \textbf{// $\alpha$ and $\beta$ are cached and non-cached sizes of the request}
            \State $node.Frequency \gets node.Frequency + 1$ 
            \If{$is\_cached$ is false}
                \State \textbf{// Bilinear interpolation to estimate the cost}
                \State Find $\alpha_l < \alpha < \alpha_h$ and $\beta_l < \beta < \beta_h$ from the profiler
                \State $T_l \gets T(\alpha_l, \beta_l) + \frac{\alpha-\alpha_l}{\alpha_h-\alpha_l} \cdot [T(\alpha_h, \beta_l) - T(\alpha_l, \beta_l)]$
                \State $T_h \gets T(\alpha_l, \beta_h) + \frac{\alpha-\alpha_l}{\alpha_h-\alpha_l} \cdot [T(\alpha_h, \beta_h) - T(\alpha_l, \beta_h)]$
                \State $T(\alpha, \beta) \gets T_l + \frac{\beta - \beta_l}{\beta_h - \beta_l} \cdot (T_h - T_l)$
                \State $node.TotalCost \gets node.TotalCost + \frac{T(\alpha, \beta)}{\beta}$
                \State $node.numComputed \gets node.numComputed + 1$
                \State $node.AvgCost \gets \frac{node.TotalCost}{node.numComputed}$
            \EndIf
            \State $node.Priority \gets Clock + node.AvgCost \times node.Frequency$
        \EndFunction
        \\
        \Function{EVICT\_IN\_GPU}{$required\_size$}
            \State $E \gets \emptyset$  \textbf{// Evicted nodes in GPU}
            \State $S \gets \{n \in GPU \wedge n.Children \notin GPU\}$  \textbf{// Leaf nodes in GPU}
            \While{$\sum_{n \in E} n.Size < required\_size$}
                \State $n \gets \argmin_{n \in S} n.Priority$
                \State $E \gets E \cup \{n\}$
                \State $Clock \gets \max\{Clock, n.Priority\}$
                \If {$n.Parent.Children \notin GPU$}
                    \State $S \gets S \cup \{n.Parent\}$
                \EndIf
            \EndWhile
        \EndFunction
    \end{algorithmic}
    \end{footnotesize}
\end{algorithm}

Algorithm~\ref{alg:design:tree} outlines the operations for updating and evicting nodes in the GPU memory 
of the knowledge tree. $T(\alpha, \beta)$ represents the estimated compute time for a request with 
$\alpha$ cached tokens and $\beta$ non-cached tokens. Upon a document retrieval by a request, 
\sysname updates the cost using bilinear interpolation (line 6--9) if the document is not cached,
\textsc{EVICT\_IN\_GPU} evicts nodes from the GPU memory to accommodate new requests and 
updates the clock according to Formula~\ref{formula:GDSF2}. If a parent node becomes a leaf following eviction, 
it is added to the candidate set $S$.

\parabf{Swap out only once.} The GPU connects to the host memory via the PCIe bus, which typically offers
much lower bandwidth than the GPU HBM. To minimize data transfer between the GPU and host memory,
\sysname adopts a swap-out-only-once strategy as shown in Figure~\ref{fig:design:tree}. 
The key-value tensors of a node are swapped out to the host memory only for the first eviction. 
The host memory keeps the key-value tensors until the node is evicted from the whole cache.
For subsequent evictions in GPU memory, \sysname directly frees the node with zero data copy. 
Given that the capacity of the host memory is one or two orders of magnitude larger than the GPU memory, 
keeping one copy of the key-value tensors in the host memory is acceptable.

\begin{figure}[t]
    \centering
    \includegraphics[width=\linewidth]{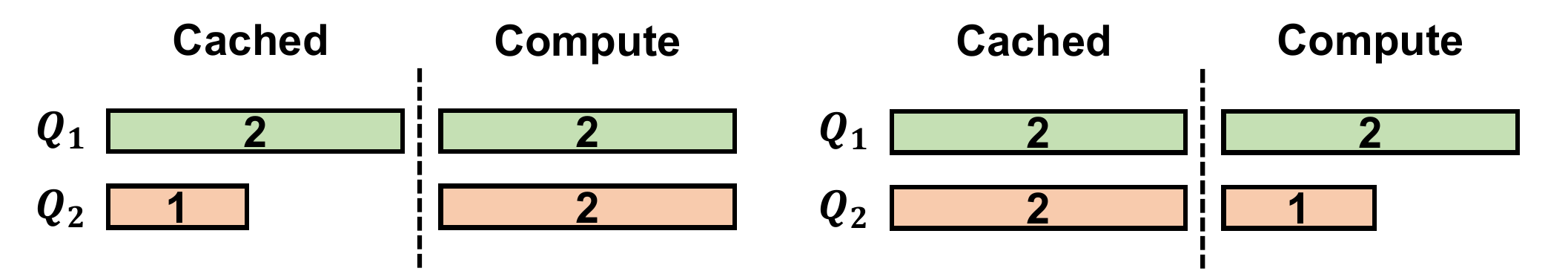}
    \newline
    \hspace*{0.5em}
        \begin{subfigure}{0.47\linewidth}
            \caption{Differ in cached length.}
            \label{fig:design:reorder:cache}
        \end{subfigure}
        \begin{subfigure}{0.47\linewidth}
            \caption{Differ in compute length.}
            \label{fig:design:reorder:compute}
        \end{subfigure}
    \vspace{-0.1in}
    \caption{Cache-aware reordering.}
    \vspace{-0.2in}
    \label{fig:design:reorder}
\end{figure}

\subsection{Cache-aware Reordering}
\label{sec:design:reorder}

Cache hit rate is vital for \sysname's cache efficiency,
yet the unpredictable arrival pattern of user requests results in substantial cache trashing.
The requests referring to the same documents may not be issued together, affecting 
cache efficiency. For illustration, let requests $\{Q_i, i\%2==0\}$ and $\{Q_i, i\%2==1\}$ target 
documents $D_1$ and $D_2$, respectively. The cache capacity is one document.
The sequence $\{Q_1, Q_2, Q_3 ...\}$ causes frequent swapping of the key-value cache of $D_1$ and $D_2$, 
rendering a zero cache hit rate. Conversely, rearranging requests to $\{Q_1, Q_3, Q_5, Q_2, Q_4, Q_6, Q_7, \ldots\}$ optimizes cache utilization, 
which improves the hit rate to 66\%. This exemplifies how strategic request ordering can mitigate cache volatility and enhance cache efficiency.

Before introducing the cache-aware reordering algorithm, we first consider two scenarios to illustrate the key insights.
We assume that the recomputation cost is proportional to the recomputation length in this example.
The first scenario (Figure~\ref{fig:design:reorder:cache}) considers requests with identical recomputation 
demands but varying cached context lengths, under a cache limit of four. With an initial order of $\{Q_1, Q_2\}$, the 
system must clear $Q_2$'s cache space to accommodate $Q_1$'s computation, then reallocate memory for $Q_1$'s 
processing. It effectively utilizes $Q_1$'s cache while discarding $Q_2$'s. This results in a total computation cost of $2 + 1 + 2 = 5$. Conversely, ordering 
as ${Q_2, Q_1}$ utilizes $Q_2$'s cache but discards $Q_1$'s, which increases computation to $2 + 2 + 2 = 6$. 
Thus, cache-aware reordering advocates prioritizing requests with larger cached contexts to enhance cache efficiency,
as they bring larger benefits.

In the second scenario, we examine requests with identical cached context lengths but varying recomputation demands,
given a cache capacity of five. For a sequence $\{Q_1, Q_2\}$, the system must clear $Q_2$'s cache to allocate 
space for $Q_1$'s computation, given only one available memory slot. This necessitates recomputing $Q_2$ entirely,
resulting in a computation cost of $2 + 2 + 1 = 5$. In contrast, the sequence $\{Q_2, Q_1\}$ allows for direct computation
of $Q_2$ due to adequate cache availability. It reduces the total computation to $2 + 1 = 3$. Hence, cache-aware 
reordering is beneficial when it prioritizes requests with shorter recomputation segments, 
as this approach minimizes the adverse effects on cache efficiency.

Drawing from these insights, we introduce a cache-aware reordering algorithm aimed at improving cache efficiency.
\sysname employs a priority queue for managing incoming requests, prioritizing them based on their impact on 
cache performance. Specifically, requests are selected for processing based on a priority metric, defined as:
\begin{align}
\label{formula:priority}
\nonumber
& \mathit{OrderPriority} = \frac{\mathit{Cached\ Length}}{\mathit{Computation\ Length}}
\end{align}
This formula prioritizes requests that are likely to enhance cache efficiency—--those with a larger cached portion 
relative to their computation needs.
By adopting this cache-aware reordering, \sysname increases the cache hit rate and decreases the total computation time,
optimizing resource use and system performance.
To avoid starvation, \sysname sets a window for each request to ensure that all requests are processed no later than the window size. 

\subsection{Dynamic Speculative Pipelining}
\label{sec:design:speculative}

As we discuss in \S\ref{sec:characterization:bottleneck}, the LLM generation is the
key performance bottleneck in RAG systems. However, if the vector database grows to a
larger scale or the retrieving requires a higher accuracy, the retrieval step may
incur a substantial latency. To mitigate the impact of retrieval latency, \sysname
employs dynamic speculative pipelining to overlap knowledge retrieval and LLM inference.
The key insight behind this technique is that the vector search may produce
the final results early in the retrieval step, which can be leveraged
by LLM for speculative generation ahead of time.

\begin{figure}[t]
    \centering
    \includegraphics[width=0.9\linewidth]{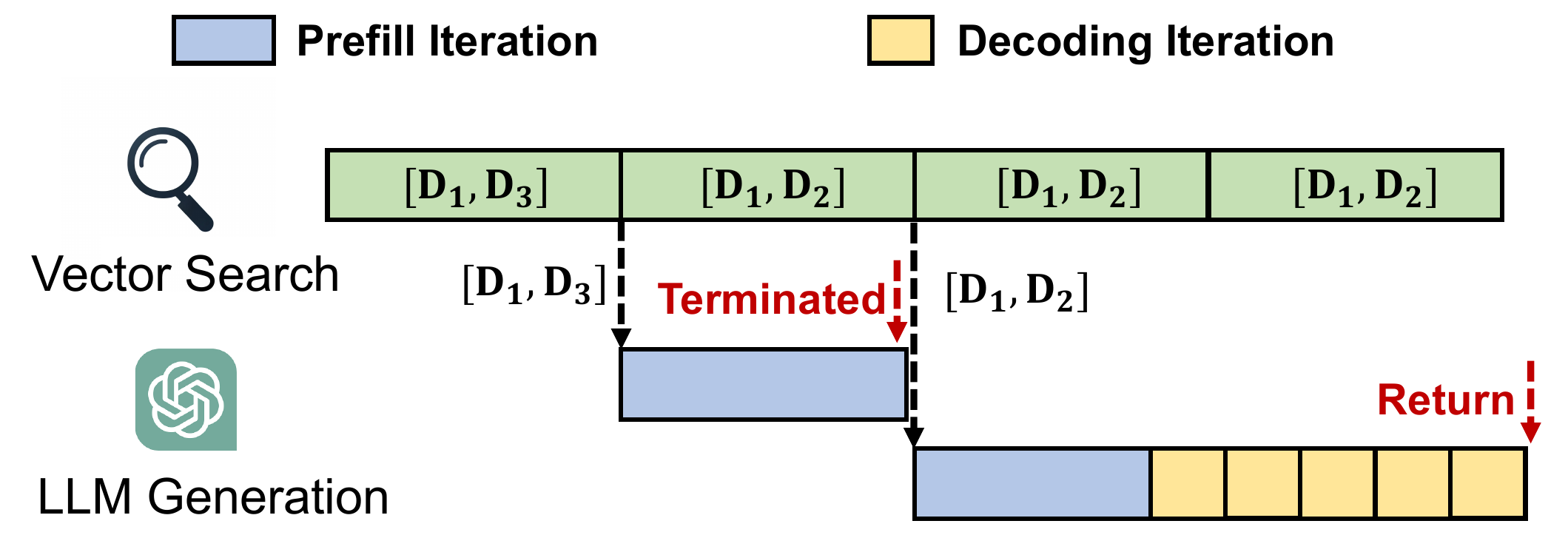}
    \vspace{-0.1in}
    \caption{Speculative pipelining.}
    \vspace{-0.2in}
    \label{fig:design:pipeline}
\end{figure}

Specifically, vector search maintains a queue of top-$k$ candidate documents,
which are ranked by their similarity to the request.
During the retrieval process, the top-$k$ documents in the queue are continuously updated
i.e., some documents with greater similarity are inserted into the queue.
However, the final top-$k$ documents may emerge early in the retrieval step~\cite{zhang2023fast, li2020improving}.
Based this observation, \sysname
introduces a speculative pipelining strategy that splits a request's retrieval process into several
stages. In each stage, \sysname ticks the vector database to send the candidate documents to the LLM
engine for speculative generation. Then, the LLM engine starts a new speculative generation
and terminates the previous generation if the received documents are different from the previous ones.
If there is no difference, the LLM engine remains processing the previous generation.
When the final top-$k$ documents are produced, \sysname sends the final results
to the LLM engine. At this moment, the LLM engine returns the results of the latest speculative generation 
to users if it matches the final top-$k$ documents. Otherwise, the LLM engine performs re-generation.

As shown in Figure~\ref{fig:design:pipeline}, we split the retrieval process into four stages.
The top-$2$ documents in candidate queue are $[D_1, D_3]$, $[D_1, D_2]$, $[D_1, D_2]$, and $[D_1, D_2]$
in the four stages.
After stage one is finished, \sysname sends $[D_1, D_3]$ to the LLM engine for speculative generation.
When stage two is finished, \sysname sends $[D_1, D_2]$ to the LLM engine.
The LLM engine finds that $[D_1, D_3]$ are different from $[D_1, D_2]$, and thus terminates the
previous speculative generation and starts a new one.
As for stage three, the LLM engine receives the same documents as stage two, and
thus remains processing the previous generation.
After the final stage, \sysname sends the final top-$2$ documents to the LLM engine
which is the same as the latest speculative generation. The LLM engine directly returns the
speculative generation results to users.

\begin{figure}[t]
    \centering
    \includegraphics[width=0.8\linewidth]{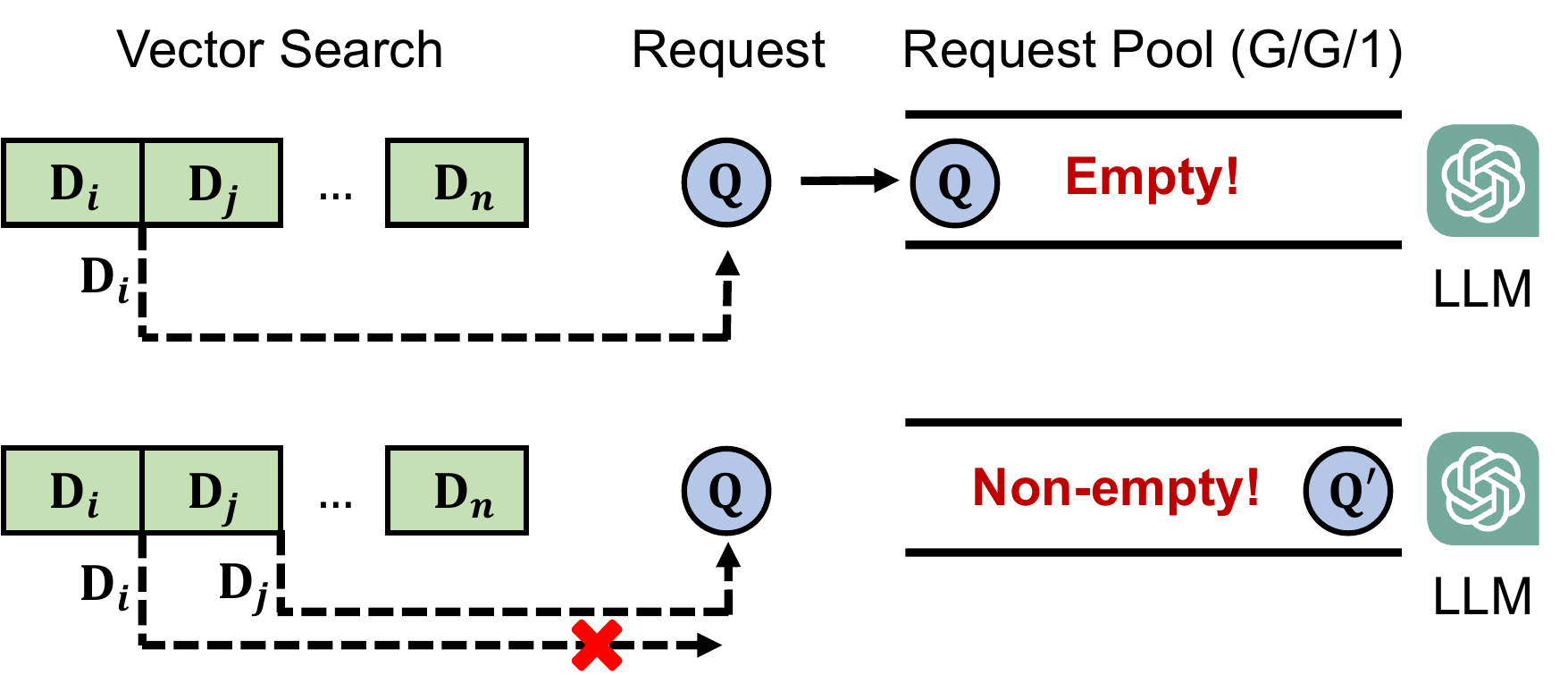}
    \vspace{-0.1in}
    \caption{The optimal speculative pipelining strategy in the simplified analysis.}
    \vspace{-0.2in}
    \label{fig:design:dynamic_speculative}
\end{figure}

\parabf{Dynamic speculative generation.} The speculative pipelining allows \sysname to overlap the retrieval and generation steps,
which reduces the end-to-end latency of RAG systems. However, it may 
introduce extra LLM computation as some speculative generations are incorrect, potentially leading to 
performance degradation under high system loads. To address this problem, \sysname dynamically 
enables speculative pipelining based on the system load.
We begin with a simplified analysis to demonstrate how to minimize the end-to-end latency while 
keeping the system load under control.

For simplicity, we assume the vector search and the LLM both serve one request at a time. 
The vector search produces candidate retrieval results at the end of each stage with a fixed time interval $d$.
Since the batch size is one, we can immediately terminate any incorrect speculative generation request.
We model the LLM engine with a G/G/1 queue~\cite{shortle2018fundamentals}. 
Since priority scheduling is widely used in LLM serving (e.g., MLFQ in FastServe~\cite{wu2023fast} and 
cache-aware reordering in \sysname), we assume that the LLM engine can schedule the requests in the queue 
with arbitrary order but executes the speculation generation requests from a single request in order. 
Figure~\ref{fig:design:dynamic_speculative} shows the optimal speculative pipelining strategy under this setting, 
which we summarize as a theorem below.
\begin{theorem}
\label{theorem:speculative}
Let $d$ be the speculative time interval and $T$ be the serving time distribution for the LLM in the G/G/1 
request pool, satisfying $T \geq d$ for all serving events. The optimal strategy 
is: start a speculative generation if the request pool is empty after a stage ends with a 
different document sequence; if not, continue vector search until the pool empties.
\end{theorem}

\begin{proof}
    After the $i$-th stage, the vector search produces a 
    document sequence $D_i$. There are four cases.

    (1) The pool is empty and $D_i$ is the final result. In this case, deferring the speculative generation
    incurs additional waiting time, while starting the speculative generation immediately utilizes the LLM
    generation engine. As a result, starting a speculative generation leads to a lower latency.

    (2) The pool is non-empty and $D_i$ is not the final result. In this case, inserting the speculative 
    generation may delay the request in the pool due to the reordering policy. Furthermore, computing the 
    speculative generation is unnecessary. In such case, deferring the speculative generation is more efficient.

    (3) The pool is non-empty and $D_i$ is the final result. In this case, inserting the speculative
    generation into the LLM request pool allows the LLM engine scheduler to conduct more efficient scheduling 
    optimizations. It is more efficient to start the speculative generation immediately.

    (4) The pool is empty and $D_i$ is not the final result. The incorrect speculative generation is
    terminated immediately upon vector search produces the first differing result. Though the speculative generation is
    incorrect, it utilizes only the idle GPU resources. Both deferring and starting the speculative generation
    are equally efficient.

    This concludes the proof.
\end{proof}

\begin{algorithm}[t]
    \caption{Dynamic Speculative Pipelining Strategy}
    \label{alg:design:speculative}
    \begin{footnotesize}
    \begin{algorithmic}[1]
        \Function{DYNAMIC\_SPECULATIVE\_PIPELINING}{$request$}
            \State $D \gets []$
            \While {the vector search of $request$ is not finished} 
                \State \textbf{// Produce the candidate documents at the next stage}
                \State $D_{temp} \gets $ \Call{VECTOR\_SEARCH}{$request$, $D$} 
                \If {$D_{temp} \neq D$}
                    \If {$\{request, D\}$ in $pool$}
                        \State Terminate $\{request, D\}$ after the current iteration
                    \EndIf
                    \If {$pool.size < max\_prefill\_bs$}
                        \State $pool.insert(\{request, D_{temp}\})$
                    \EndIf
                \EndIf
                \State $D \gets D_{temp}$
            \EndWhile
        \EndFunction
    \end{algorithmic}
    \end{footnotesize}
\end{algorithm}

The general RAG system is more complex with larger LLM batch sizes and parallel vector search. In addition, 
we cannot terminate a speculative generation immediately when the request is batched with other requests. 
Based on Theorem~\ref{theorem:speculative}, we design a dynamic speculative pipelining strategy in Algorithm~\ref{alg:design:speculative}. The main idea is to start a speculative generation only if the 
retrieved documents change and the number of pending LLM requests falls below a predetermined maximum batch size 
for the prefill iteration ($max\_prefill\_bs$). The maximum prefill batch size is determined by the smaller 
number of tokens that can either fit within the GPU memory or fully utilize the SMs. The strategy terminates 
the incorrect speculative generation after the current LLM iteration, which does not affect other requests in the 
batch. This strategy overlaps the retrieval and generation steps as much as possible while gauging the system load.
\section{Implementation}
\label{sec:implementation}

\begin{table}[t!]
    \centering
    \arrayrulewidth=0.5pt
    \extrarowheight=1pt
    \resizebox{\linewidth}{!}{
        \begin{tabular}{cccccc}
            \arrayrulecolor{black}\hline
            \arrayrulecolor{black}\hline
            \textbf{Model} & \textbf{Layers} & \textbf{Q/KV Heads} & \textbf{MoE} & \textbf{Model Size} & \textbf{KV Size} \\ 
            \hline
            Mistral-7B & 32 & 32/8 & no & 14 GiB & 0.125 MiB/token \\ 
            LLaMA2-7B & 32 & 32/32 & no & 14 GiB & 0.5 MiB/token \\ 
            Mixtral-8$\times$7B & 32 & 32/8 & yes & 96.8 GiB & 0.125 MiB/token \\
            LLaMA2-70B & 80 & 64/8 & no & 140 GiB & 0.3125 MiB/token \\ 
            \arrayrulecolor{black}\hline
            \arrayrulecolor{black}\hline
        \end{tabular}
    }
    \vspace{-0.05in}
    \caption{Models used in the evaluation.}
    \vspace{-0.25in}
    \label{tab:evaluation:models}
\end{table}

We implement a system prototype of \sysname with $\sim$5000 lines of code in C++ and Python.
Our implementation is based on vLLM~\cite{vllm} v0.3.0, a state-of-the-art LLM serving system.
We extend its prefill kernel in Pytorch~\cite{paszke2019pytorch} and Triton~\cite{tillet2019triton} to support 
prefix caching for different attention mechanisms, e.g., multi-head attention~\cite{vaswani2017attention} and 
grouped-query attention~\cite{ainslie2023gqa}.

\parabf{Pipelined vector search.} We implement dynamic speculative pipelining on top of 
Faiss~\cite{pinecone-faiss}, an open-source widely-used vector database, and adapt it for 
two types of indexes: IVF~\cite{babenko2014inverted} and HNSW~\cite{malkov2018efficient}.
IVF divides the vector space into multiple clusters and stores the vectors in the corresponding clusters.
During vector search, IVF first locates the top-$n$ closest clusters to the request vector and subsequently searches within these clusters.
HNSW constructs multi-level graphs to map the vector space, connecting vectors with edges to represent similarity.
HNSW searches the vectors by traversing the graph and maintains a candidate list to store the current top-$k$ nearest vectors.
To support pipelined vector search, we modify the search process of these two indexes.
Specifically, we split the IVF search into multiple stages, each searching the vectors in some clusters and 
returning the current top-$k$ vectors.
For HNSW, we measure the average search time for a specified search configuration and split the
entire time into smaller time slices. After each time slice of searching, it returns the current top-$k$ vectors.
These modifications aim to facilitate the dynamic speculative pipeline while maintaining the final search semantics.

\parabf{Fault tolerance.} We implement two fault-tolerant mechanisms in \sysname to handle GPU failures and 
request processing failures. The GPU memory serves as \sysname's first-level cache, storing the 
KV cache of the upper-level nodes in the knowledge tree hierarchy. Given the prefix sensitivity of LLM 
inference, a GPU failure would invalidate the lower-level nodes and therefore the entire tree. We replicate 
a portion of the most frequently accessed upper-level nodes (e.g., the system prompt) 
in the host memory for fast recovery. We also employ a timeout mechanism to 
retry the failed requests. If a request fails before completing its first iteration, it will be 
recomputed. Otherwise, the request can continue computation by reusing the stored KV cache.

\section{Evaluation}
\label{sec:evaluation}

In this section, we evaluate \sysname from the following aspects: $(i)$ overall performance 
against state-of-the-art approaches; $(ii)$ performance under general settings; $(iii)$ 
ablation studies on the techniques used in \sysname; and $(iv)$ scheduling time of \sysname. 

\parabf{Testbed.} Most of our experiments are conducted on AWS EC2 g5.16xlarge instances, each 
with 64 vCPUs (AMD EPYC 7R32), 256 GiB host memory, and 25 Gbps NIC.
Each instance is configured with one NVIDIA A10G GPU with 24 GiB memory and 
the GPU is connected to the host via PCIe 4.0$\times$16. 
We run experiments with 7B models on a single g5.16xlarge instance and 
use 192 GiB host memory for caching unless otherwise stated. 
For large models, we use two NVIDIA H800 GPUs, each with 80 GiB memory and 
interconnected by NVLink. The two GPUs are connected to the host via PCIe 5.0$\times$16. 
We use 384 GiB host memory for caching in this case.

\parabf{Models.} We evaluate \sysname with the LLaMA 2 chat models~\cite{touvron2023llama} 
and the Mistral AI models~\cite{jiang2023mistral, jiang2024mixtral}. 
The model details are listed in Table~\ref{tab:evaluation:models}. Most of the experiments are conducted 
with Mistral-7B and LLaMA2-7B. The two models have the same size but employ different attention mechanisms, 
i.e., grouped-query attention and multi-head attention. 
We also evaluate \sysname with large models, Mixtral-8$\times$7B and LLaMA2-70B, to demonstrate the 
scalability of \sysname. Mixtral-8$\times$7B is a mixture-of-experts (MoE) model with eight experts, 
and two experts are activated for each token. We deploy the large models on two H800 80GB GPUs for tensor 
parallelism and expert parallelism. 

\parabf{Retrieval.} We use the Wikipedia dataset collected in \S~\ref{sec:characterization:opp} as 
the knowledge base. For vector search, we use the IVF index with 1024 clusters and set the default top-$k$ to 2. 
We deploy the vector database with four separate vCPUs and 30 GiB host memory on 
the same instance as the GPU and expose a RESTful API for document retrieval.

\parabf{Workloads.} Our evaluation uses two representative QA datasets, MMLU~\cite{hendrycks2020measuring} 
and Natural Questions~\cite{kwiatkowski2019natural}. MMLU is a multi-choice knowledge benchmark 
where the LLM outputs a single token (A/B/C/D) per question. Natural Questions comprises 
anonymized questions from Google Search and provides the reference answers for each question. 
For Natural Questions, we sample the output length for each question from the token length distribution 
of the reference answers. The average output length is 6 tokens, and 99\% of the answers contain no more 
than 32 tokens. We sample a subset of the questions from the dataset, 
respecting the document retrieval distribution in \S~\ref{sec:characterization:opp}, and 
randomly shuffle the questions to generate 1-hour workloads.
In line with prior work~\cite{wu2023fast, vllm}, we assign the arrival time for each request 
using a Poisson process parameterized by the arrival rate. 

\parabf{Metrics.} We report the average time-to-first-token (TTFT) as the main metric. 
We also evaluate the system throughput, which is defined as the request rate that the system 
can process while maintaining a TTFT below a certain threshold, e.g., 5$\times$ of the TTFT 
at the lowest request rate. In addition, we measure the cache hit rate in the ablation study.

\parabf{Baselines.} We compare \sysname with two baselines.
\begin{itemize}[leftmargin=*]
    \item \textbf{vLLM}~\cite{vllm}, a state-of-the-art LLM serving system that supports iteration-level 
    scheduling~\cite{yu2022orca} and uses PagedAttention~\cite{vllm} to reduce memory fragmentation.
    \item \textbf{SGLang}~\cite{sglang}, a high-performance LLM serving system that allows KV cache reuse 
    across different requests in GPU memory and employs LRU as the replacement policy.
\end{itemize}
For fair comparison, the baselines are configured with the same model parallelism, maximum batch size, 
and vector database settings as \sysname. 

\begin{figure}[t!]
    \centering
    \includegraphics[width=\linewidth]{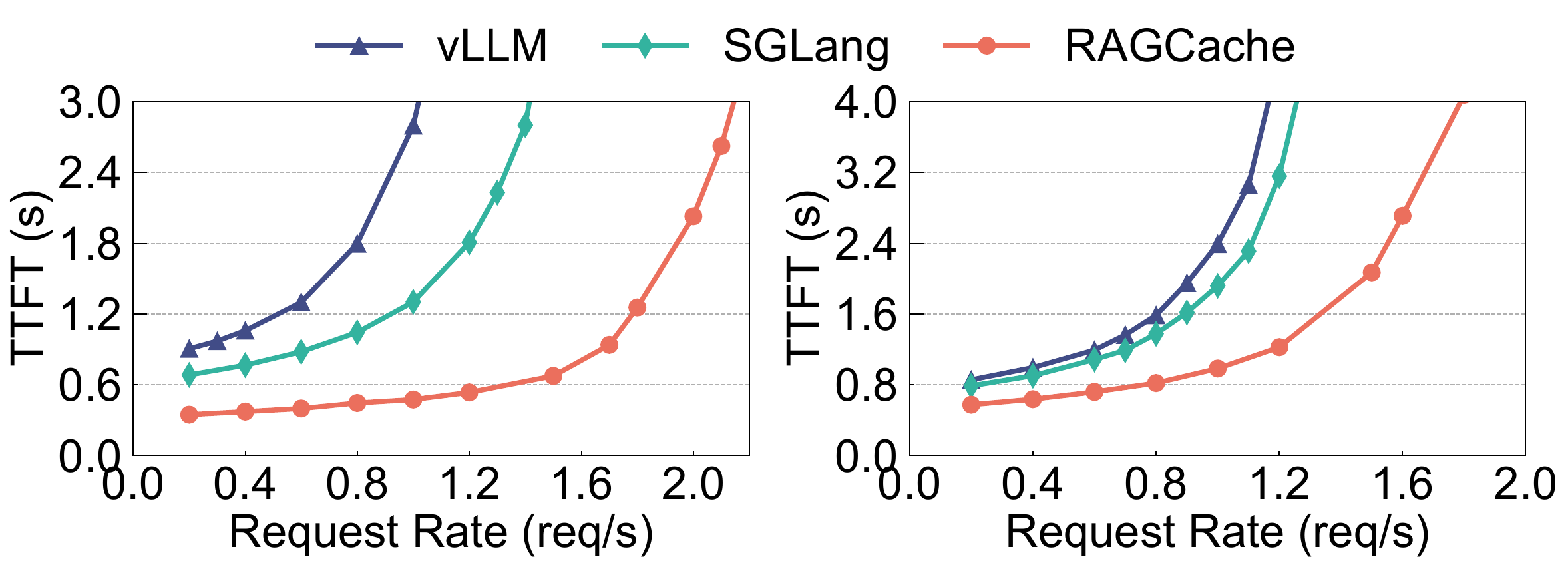}
    \vspace{-0.2in}
    \newline
    \hspace*{0.5em}
        \begin{subfigure}{0.47\linewidth}
            \caption{Mistral-7B.}
            \label{fig:overall:mmlu:mistral-7b}
        \end{subfigure}
        \begin{subfigure}{0.47\linewidth}
            \caption{LLaMA2-7B.}
            \label{fig:overall:mmlu:llama2-7b}
        \end{subfigure}
    \vspace{-0.1in}
    \caption{Overall performance on MMLU.}
    \vspace{-0.1in}
    \label{fig:overall:mmlu}
\end{figure}

\subsection{Overall Performance}
\label{sec:evaluation:overall}

We first compare the overall performance of \sysname against the baselines.
We use MMLU and Natural Questions as the workloads and Mistral-7B and LLaMA2-7B as the models. 
The maximum batch size is set to 4.
We vary the request rate and measure the average TTFT.
Figure~\ref{fig:overall:mmlu} and Figure~\ref{fig:overall:googlenq} show the results on
MMLU and Natural Questions, respectively, which we summarize as follows.
\begin{itemize}[leftmargin=*]
    \item \sysname reduces the average TTFT by 1.2--4$\times$ compared to vLLM and 1.1--3.5$\times$
    compared to SGLang under the same request rate. This is because \sysname utilizes the GPU memory 
    and host memory to cache the KV cache of hot documents and avoids frequent recomputation.
    \item Due to faster request processing, \sysname achieves 1.3--2.1$\times$ higher throughput 
    than vLLM and 1.2--1.8$\times$ higher throughput than SGLang.
    \item \sysname outperforms the baselines across models with varying attention mechanisms 
    on different datasets.
\end{itemize}
The results also reflect the differences between the two models and datasets. 
Notably, the performance gap between \sysname and vLLM is greater for Mistral-7B than for LLaMA2-7B
This is because LLaMA2-7B has a KV cache size 4$\times$ that of Mistral-7B for the same token count.
resulting in a lower cache hit rate for LLaMA2-7B with the same cache size.
According to our characterization in \S~\ref{sec:characterization:opp}, 
MMLU benefits more from document caching than Natural Questions, 
with a wider performance improvement for MMLU than for Natural Questions.
SGLang performs closely to vLLM for Natural Questions 
because the limited GPU memory restricts document locality.
\sysname, however, with its multilevel caching and the adapted knowledge tree,  
outperforms in both datasets.

\begin{figure}[t!]
    \centering
    \includegraphics[width=\linewidth]{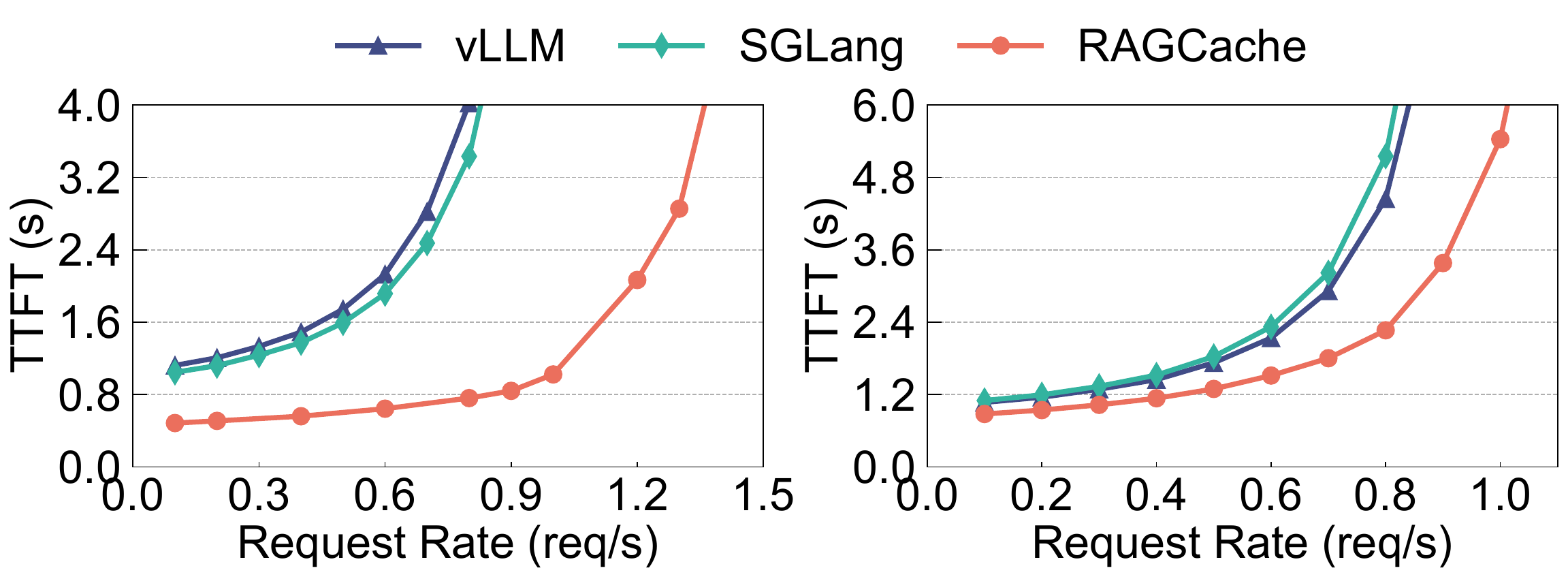}
    \vspace{-0.2in}
    \newline
    \hspace*{0.5em}
        \begin{subfigure}{0.47\linewidth}
            \caption{Mistral-7B.}
            \label{fig:overall:googlenq:mistral-7b}
        \end{subfigure}
        \begin{subfigure}{0.47\linewidth}
            \caption{LLaMA-2-7B.}
            \label{fig:overall:googlenq:llama2-7b}
        \end{subfigure}
    \vspace{-0.1in}
    \caption{Overall performance on Natural Questions.}
    \vspace{-0.1in}
    \label{fig:overall:googlenq}
\end{figure}

\begin{figure*}[t!]
    \centering
    \includegraphics[width=\linewidth]{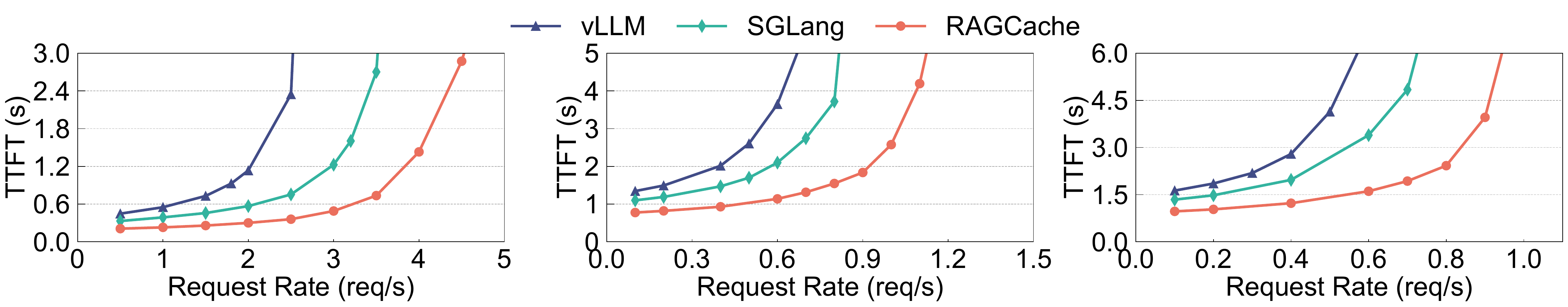}
    \vspace{-0.2in}
    \newline
    \hspace*{0.5em}
        \begin{subfigure}{0.32\linewidth}
            \caption{Top-1.}
            \label{fig:topk:mmlu:top1}
        \end{subfigure}
        \begin{subfigure}{0.32\linewidth}
            \caption{Top-3.}
            \label{fig:topk:mmlu:top3}
        \end{subfigure}
        \begin{subfigure}{0.32\linewidth}
            \caption{Top-5.}
            \label{fig:topk:mmlu:top5}
        \end{subfigure}
    \vspace{-0.1in}
    \caption{Performance with different top-$k$ values.}
    \vspace{-0.1in}
    \label{fig:topk:mmlu}
\end{figure*}

\subsection{Case Study}

We then conduct two case studies to demonstrate the benefits of \sysname
over the baselines 
under general settings. We use MMLU and Mistral-7B in the case studies.

\parabf{Different top-$k$ values.} Users may have varying requirements for the number of retrieved documents. 
We evaluate the performance of \sysname and the baselines with commonly used top-$k$ values: 1, 3, and 5.
We set the maximum batch size to 4 and truncate the documents in the top-5 experiment to fit within GPU capacity
limits. Figure~\ref{fig:topk:mmlu} shows that \sysname outperforms vLLM 
by 1.7--3.1$\times$ and SGLang by 1.2--2.5$\times$ in average TTFT across these top-$k$ values.
Despite the factorial growth in document permutations with increasing top-$k$ values, 
\sysname maintains its advantage by caching frequently-used documents. This is because the knowledge tree 
always evicts the node furthest from the root, ensuring that the most frequently used prefixes 
remain in the cache.

\begin{figure}[t!]
    \centering
    \includegraphics[width=\linewidth]{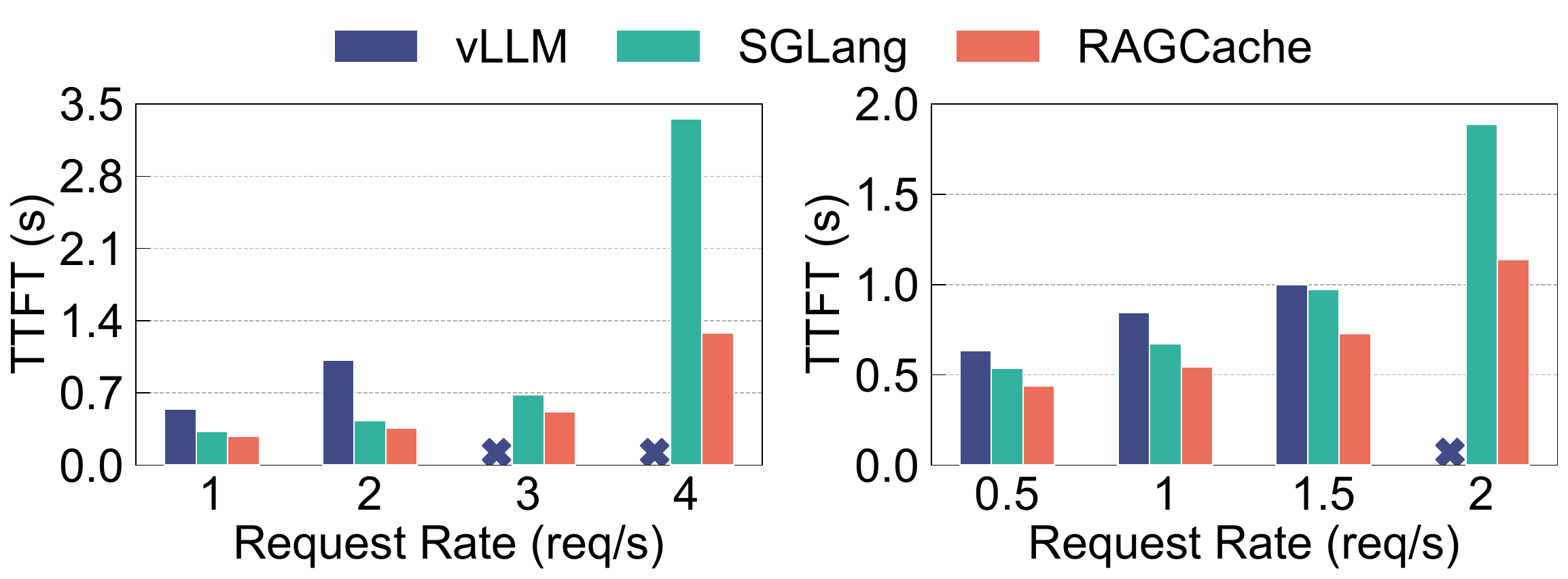}
    \vspace{-0.2in}
    \newline
    \hspace*{0.5em}
        \begin{subfigure}{0.47\linewidth}
            \caption{Mixtral-8$\times$7B.}
            \label{fig:case:large:mixtral-8x7b}
        \end{subfigure}
        \begin{subfigure}{0.47\linewidth}
            \caption{LLaMA2-70B.}
            \label{fig:case:large:llama2-70b}
        \end{subfigure}
    \vspace{-0.1in}
    \caption{Performance under large models.}
    \vspace{-0.1in}
    \label{fig:case:large}
\end{figure}

\parabf{Large models.} The second case study evaluates \sysname with larger models using MMLU as the workload.
We deploy Mixtral-8$\times$7B and LLaMA2-70B on two H800 80GB GPUs. 
For each model, we set the maximum batch size to the lesser of what fits in GPU memory or fully utilizes the SMs 
(e.g., 8 for Mixtral-8$\times$7B and 4 for LLaMA2-70B). 
With a limited budget, we run the experiments with four different request rates per model and 
set a TTFT SLO at 5$\times$ the TTFT of the lowest request rate. 
Figure~\ref{fig:case:large} shows that under low request rates, \sysname reduces the average 
TTFT by 1.4--2.1$\times$ compared to vLLM. vLLM fails to meet the SLO above 2 and 1.5 requests per second for 
Mixtral-8$\times$7B and LLaMA2-70B, respectively, whereas \sysname 
maintains the TTFT below 1.4 seconds across varying request rates. 
SGLang performs better on H800 than on A10G GPUs due to increased GPU memory for caching, but 
\sysname still surpasses SGlang by 1.2--2.6$\times$ in average TTFT.

\subsection{Ablation Study}
\label{sec:evaluation:ablation}

\paraf{Prefix-aware GDSF policy.} 
We compare \sysname with versions of \sysname that use native GDSF, LRU, and LFU as the replacement policy.
For the GDSF policy, we set the recomputation cost of a document proportional to the document size, 
which aligns with our profiling results in Figure~\ref{fig:characterization:inference}.
We vary the host memory size for caching from 8 GiB to 128 GiB, set the request rate to 0.8 req/s, 
and report the hit rate and average TTFT. The hit rate for top-2 retrieval is defined as 
the ratio of the number of hit documents to the number of retrieved documents. For example, 
if the stored document sequence is $[D_1, D_2]$ and the requested one is $[D_1, D_3]$, then 
the hit rate would be 50\%.
Figure~\ref{fig:ablation:policy} shows the hit rate for MMLU and Natural Questions, 
and Table~\ref{tab:evaluation:ablation:policy} lists the corresponding average TTFT. 
PGDSF achieves the highest hit rate across different host memory sizes, 
with a 1.02--1.32$\times$ improvement over GDSF, 1.06--1.62$\times$ improvement over LRU, and 
1.06--1.75$\times$ improvement over LFU. 
This is because PGDSF captures the varying sizes, access patterns, and recomputation costs of different 
document prefixes. 
With the higher hit rate, \sysname achieves 1.05--1.29$\times$ lower average TTFT than the 
baseline policies.

\parabf{Cache-aware reordering.} Then we evaluate the impact of cache-aware reordering. 
Reordering works when the request queue is saturated. We set the request rate to 2.5 req/s for MMLU and 
1.4 req/s for Natural Questions, which are slightly higher than the throughput of \sysname. 
We set the reordering window size to 32 and vary the host memory size from 16 GiB to 128 GiB.
Figure~\ref{fig:ablation:reorder} shows that \sysname 
reduces the average TTFT by 1.2--2.1$\times$ with cache-aware reordering, 
which demonstrates its effectiveness under high request rates.

\begin{figure}[t!]
    \centering
    \includegraphics[width=\linewidth]{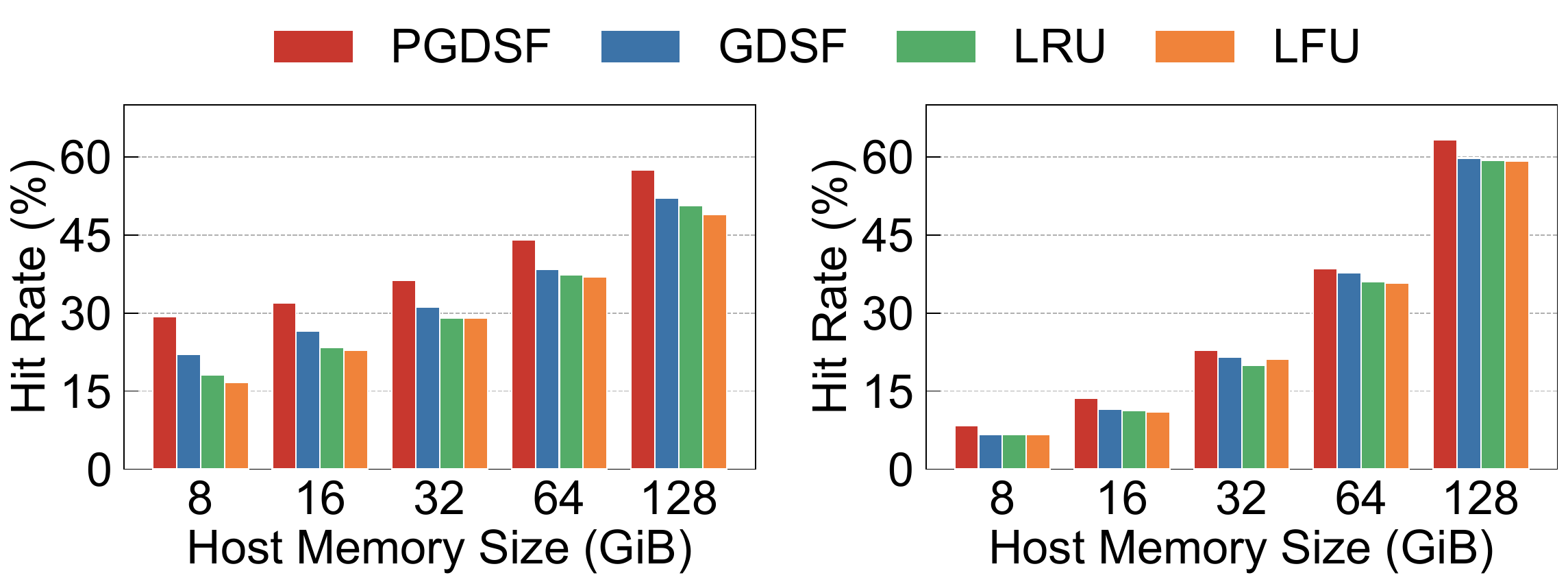}
    \vspace{-0.2in}
    \newline
    \hspace*{0.5em}
        \begin{subfigure}{0.47\linewidth}
            \caption{MMLU.}
            \label{fig:ablation:policy:mmlu}
        \end{subfigure}
        \begin{subfigure}{0.47\linewidth}
            \caption{Natural Questions.}
            \label{fig:ablation:policy:googlenq}
        \end{subfigure}
    \vspace{-0.1in}
    \caption{Ablation study on cache replacement policy.}
    \vspace{-0.1in}
    \label{fig:ablation:policy}
\end{figure}

\begin{table}[t!]
    \centering
    \arrayrulewidth=0.5pt
    \extrarowheight=1pt
    \resizebox{\linewidth}{!} {
        \begin{tabular}{c|c|c|c|c|c|c|c|c}
            \arrayrulecolor{black}\hline
            \arrayrulecolor{black}\hline
            \multirow{2}{*}{\begin{tabular}[c]{@{}c@{}}Host Memory\\Size\end{tabular}} & \multicolumn{4}{c|}{MMLU} & \multicolumn{4}{c}{Natural Questions}\\ 
            \cline{2-9} 
            & \textbf{PGDSF} & \textbf{GDSF} & \textbf{LRU} & \textbf{LFU} & \textbf{PGDSF} & \textbf{GDSF} & \textbf{LRU} & \textbf{LFU} \\ 
            \hline
            8 GiB & 1.38 & 1.68 & 1.78 & 1.81 & 2.85 & 3.35 & 3.41 & 3.36 \\ 
            16 GiB & 1.32 & 1.55 & 1.61 & 1.63 & 2.50 & 2.89 & 2.92 & 2.98 \\ 
            32 GiB & 1.23 & 1.45 & 1.50 & 1.49 & 2.00 & 2.09 & 2.25 & 2.20 \\ 
            64 GiB & 1.06 & 1.27 & 1.28 & 1.29 & 1.32 & 1.47 & 1.56 & 1.55 \\
            128 GiB & 0.83 & 0.98 & 1.01 & 1.03 & 0.78 & 0.92 & 0.95 & 0.95 \\
            \arrayrulecolor{black}\hline
            \arrayrulecolor{black}\hline
        \end{tabular}
    }
    \vspace{0.05in}
    \caption{Average TTFT (seconds) of different replacement policies with varying host memory size.}
    \vspace{-0.2in}
    \label{tab:evaluation:ablation:policy}
\end{table}

\parabf{Dynamic speculative pipelining.} Finally, we evaluate the effectiveness of dynamic 
speculative pipelining against a baseline, No Dynamic Speculative Pipelining (No DSP), 
which waits for vector search completion before starting LLM generation.
We vary the ratio of the number of searched vectors to the total number of vectors from 12.5\% to 100\%. 
Note that while the search accuracy increases with the vector search ratio, the search time also extends. 
We use MMLU and Natural Questions as the workloads and set the request rate to 0.1 req/s. 
Figure~\ref{fig:ablation:speculative} demonstrates that \sysname achieves up to 1.6$\times$ 
TTFT reduction with dynamic speculative pipelining. Table~\ref{tab:evaluation:ablation:speculative} presents the 
average non-overlapping vector search time, which refers to the duration that the vector search 
does not overlap with the LLM generation using the final retrieval result.
Dynamic speculative pipelining allows \sysname to decrease non-overlapping vector search time by 
1.5--4.3$\times$ and leads to a lower TTFT.

\subsection{Scheduling Time}
\label{sec:evaluation:overhead}

We measure \sysname's scheduling time, including the time for knowledge tree lookup and update, 
request reordering, and speculative pipelining decisions.
Using MMLU as the workload and Mistral-7B as the model, we range the request rate 
from 0.5 to 2 req/s. Table~\ref{tab:evaluation:scheduling_time} indicates that the scheduling time 
remains below one millisecond across all request rates, which is negligible compared to the second-level TTFT.

\begin{figure}[t!]
    \centering
    \includegraphics[width=\linewidth]{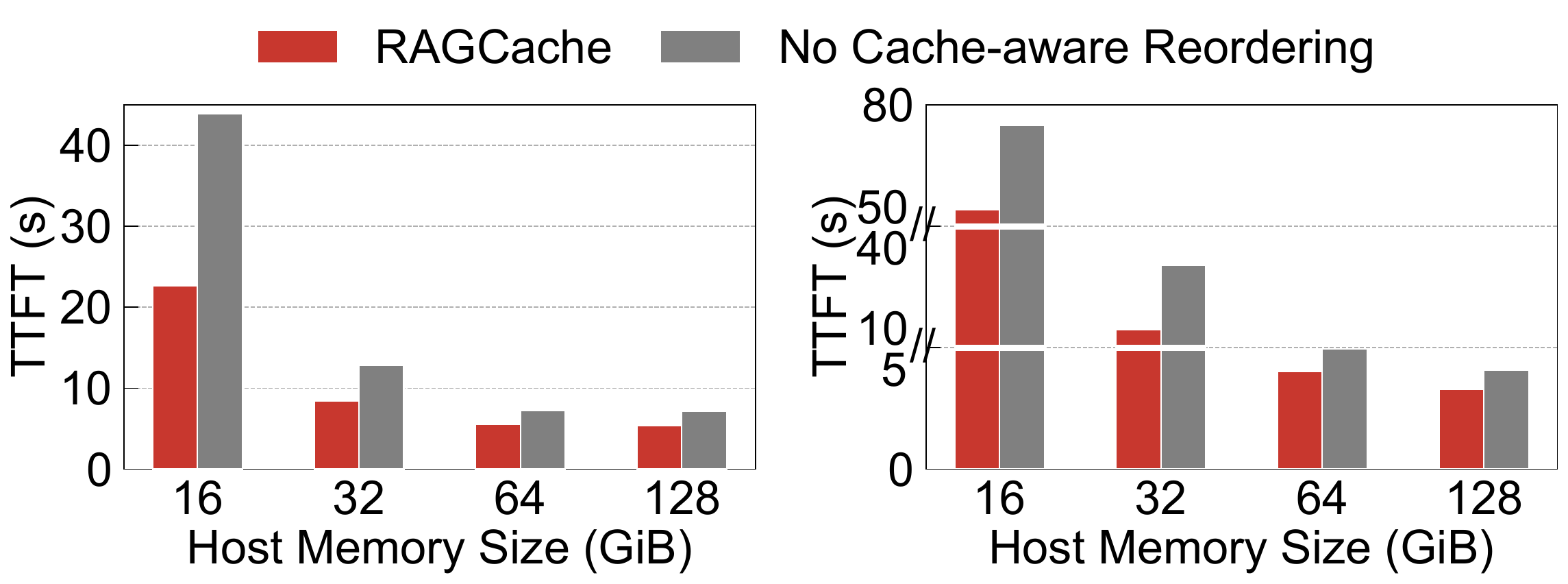}
    \vspace{-0.2in}
    \newline
    \hspace*{0.5em}
        \begin{subfigure}{0.47\linewidth}
            \caption{MMLU.}
            \label{fig:ablation:reorder:mmlu}
        \end{subfigure}
        \begin{subfigure}{0.47\linewidth}
            \caption{Natural Questions.}
            \label{fig:ablation:reorder:googlenq}
        \end{subfigure}
    \vspace{-0.1in}
    \caption{Ablation study on cache-aware reordering.}
    \vspace{-0.1in}
    \label{fig:ablation:reorder}
\end{figure}

\begin{figure}[t!]
    \centering
    \includegraphics[width=\linewidth]{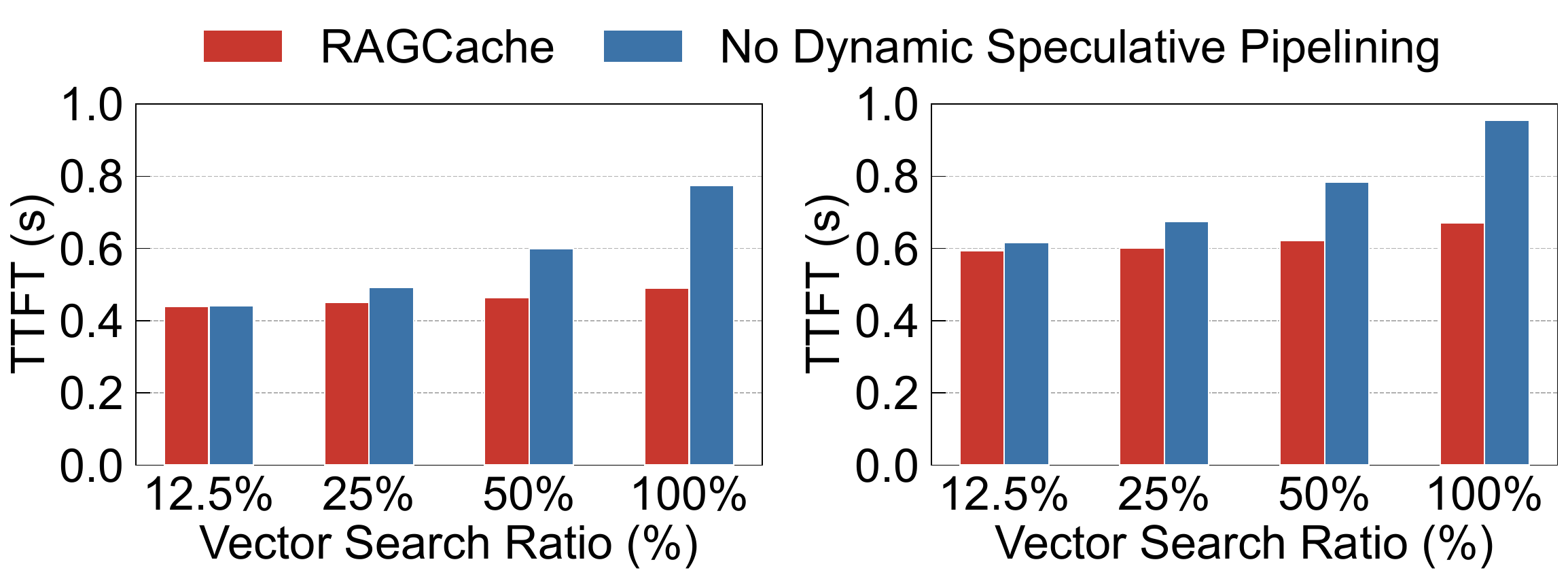}
    \vspace{-0.2in}
    \newline
    \hspace*{0.5em}
        \begin{subfigure}{0.47\linewidth}
            \caption{MMLU.}
            \label{fig:ablation:speculative:mmlu}
        \end{subfigure}
        \begin{subfigure}{0.47\linewidth}
            \caption{Natural Questions.}
            \label{fig:ablation:speculative:googlenq}
        \end{subfigure}
    \vspace{-0.1in}
    \caption{Ablation study on speculative pipelining.}
    \vspace{-0.1in}
    \label{fig:ablation:speculative}
\end{figure}

\section{Discussion}
\label{sec:discussion}

\paraf{Time per output token (TPOT).} In addition to time to first token (TTFT), 
TPOT is crucial for LLM serving~\cite{zhong2024distserve}. 
RAG augments the request with documents retrieved from the external knowledge base, which significantly increases
the input length and thus the latency of the prefill phase, i.e., TTFT. 
Consequently, the primary concern for RAG systems is the prolonged TTFT due to the extended input length. 
\sysname reduces TTFT by caching the KV cache of the most frequently retrieved documents and 
can also lower TPOT by accelerating the prefill iteration, 
as decoding iterations typically take far less time than the prefill iteration~\cite{wu2023fast}.

\parabf{Large top-$k$.} As the top-k value increases, the number of document permutations explodes in factorial growth,
making them less likely to be reused. \sysname mitigates this by caching documents with a lower top-$k$ value 
(e.g., caching the top-3 documents for requests with top-5 documents), thereby achieving a balance between 
hit rate and cache efficiency.

\section{Related Work}
\label{sec:related}

\paraf{RAG.} RAG~\cite{lewis2020retrieval, jiang2024piperag, borgeaud2022improving, ram2023context, trivedi2022interleaving, langchain} 
enhances the generation quality of LLMs by incorporating 
relevant knowledge from external databases. 
Several works~\cite{jiang2024piperag, ram2023context, trivedi2022interleaving, zhang2024accelerating} suggest 
iterative retrieval throughout generation to further improve the response quality. 
\sysname supports iterative retrieval by treating the intermediate iterations as separate requests and caching 
the corresponding KV cache of the documents.

\parabf{Vector search.} RAG systems convert user prompts into vectors and 
uses approximate nearest neighbor (ANN) indexes like IVF~\cite{babenko2014inverted, chen2021spann, zhang2023fast, zhang2024fast} 
and graph indexes~\cite{malkov2018efficient, fu2019fast, jayaram2019diskann} for efficient and accurate similarity search. 
\sysname extracts the temporary search results for speculative LLM generation and thus pipelines 
the search process with LLM inference.

\parabf{KV cache management.} KV cache is widely used to accelerate the decoding phase of LLM 
inference~\cite{vllm, yu2022orca, zhong2024distserve, wu2023fast, xiao2023efficient}. 
Recent efforts aim to reduce the KV cache's memory footprint by  
quantization~\cite{dettmers2022gpt3}, compression~\cite{li2020train, liu2024scissorhands, ge2023model}, and 
self-attention with a subset of tokens~\cite{xiao2023efficient, zhang2024h2o}. 
These methods introduce approximation to the generation process, while \sysname preserves the exact 
KV cache of documents without affecting generation quality.
Inspired by virtual memory in operating systems, vLLM~\cite{vllm} manages the KV 
cache at page granularity and proposes PagedAttention to prevent external fragmentation. 
\sysname integrates the page-level management for KV cache sharing and improves over vLLM by 
leveraging the characteristics of RAG to cache the KV cache of the knowledge documents.

\parabf{KV cache reusing.} Recent efforts~\cite{gim2023prompt, liu2023cachegen, sglang, ye2024chunkattention} 
propose to reuse the KV cache across requests to reduce redundant computation. 
Prompt Cache~\cite{gim2023prompt} allows flexible reuse of the same tokens at different positions, while
CacheGen~\cite{liu2023cachegen} compresses the KV cache for efficient reuse. 
Both approaches may generate inaccurate responses. SGLang~\cite{sglang} and ChunkAttention~\cite{ye2024chunkattention} 
identify the reusable KV cache in GPU memory. \sysname leverages RAG's retrieval pattern 
and builds a multilevel caching system, leading to higher performance with unchanged generation results.

\begin{table}[t!]
    \centering
    \arrayrulewidth=0.5pt
    \extrarowheight=1pt
    \resizebox{0.75\linewidth}{!} {
        \begin{tabular}{c|c|c|c|c}
            \arrayrulecolor{black}\hline
            \arrayrulecolor{black}\hline
            \multirow{2}{*}{\begin{tabular}[c]{@{}c@{}}Vector Search\\Ratio\end{tabular}} & \multicolumn{2}{c|}{MMLU} & \multicolumn{2}{c}{Natural Questions}\\ 
            \cline{2-5} 
            & \textbf{\sysname} & \textbf{No DSP} & \textbf{\sysname} & \textbf{No DSP} \\ 
            \hline
            12.5\% & 52.1 ms & 78.5 ms & 67.7 ms & 105.8 ms \\ 
            25\% & 59.2 ms & 135.9 ms & 72.9 ms & 163.4 ms \\ 
            50\% & 69.7 ms & 243.7 ms & 94.2 ms & 282.5 ms \\ 
            100\% & 97.4 ms & 422.3 ms & 145.0 ms & 446.1 ms \\
            \arrayrulecolor{black}\hline
            \arrayrulecolor{black}\hline
        \end{tabular}
    }
    \vspace{0.05in}
    \caption{Average non-overlapping vector search time under different settings.}
    \vspace{-0.1in}
    \label{tab:evaluation:ablation:speculative}
\end{table}

\begin{table}[t!]
    \centering
    \arrayrulewidth=0.5pt
    \extrarowheight=1pt
    \resizebox{0.45\linewidth}{!} {
        \begin{tabular}{c|c}
            \arrayrulecolor{black}\hline
            \arrayrulecolor{black}\hline
            Request Rate & Scheduling Time \\ 
            \hline
            0.5 req/s & 0.880 ms \\ 
            1.0 req/s & 0.872 ms \\
            1.5 req/s & 0.902 ms \\
            2.0 req/s & 0.906 ms \\
            \arrayrulecolor{black}\hline
            \arrayrulecolor{black}\hline
        \end{tabular}
    }
    \vspace{0.05in}
    \caption{Scheduling time of \sysname.}
    \vspace{-0.25in}
    \label{tab:evaluation:scheduling_time}
\end{table}

\section{Conclusion}
\label{sec:conclusion}

We present \sysname, a multilevel caching system tailored for RAG. Based on a detail RAG system characterization, 
\sysname employs a knowledge tree with a prefix-aware replacement policy to minimize  redundant computation 
and a dynamic speculative pipelining mechanism to overlap the knowledge retrieval and LLM inference 
in the RAG workflow.
We evaluate \sysname with a variety of models and workloads. 
The experimental results show that \sysname outperforms the state-of-the-art solution, vLLM integrated with Faiss, 
by up to 4$\times$ on TTFT and 2.1$\times$ on throughput.
j

\bibliographystyle{ACM-Reference-Format}
\bibliography{paper}


\begin{thebibliography}{00}


\ifx \showCODEN    \undefined \def \showCODEN     #1{\unskip}     \fi
\ifx \showDOI      \undefined \def \showDOI       #1{#1}\fi
\ifx \showISBNx    \undefined \def \showISBNx     #1{\unskip}     \fi
\ifx \showISBNxiii \undefined \def \showISBNxiii  #1{\unskip}     \fi
\ifx \showISSN     \undefined \def \showISSN      #1{\unskip}     \fi
\ifx \showLCCN     \undefined \def \showLCCN      #1{\unskip}     \fi
\ifx \shownote     \undefined \def \shownote      #1{#1}          \fi
\ifx \showarticletitle \undefined \def \showarticletitle #1{#1}   \fi
\ifx \showURL      \undefined \def \showURL       {\relax}        \fi
\providecommand\bibfield[2]{#2}
\providecommand\bibinfo[2]{#2}
\providecommand\natexlab[1]{#1}
\providecommand\showeprint[2][]{arXiv:#2}

\bibitem[\protect\citeauthoryear{??}{lan}{2024}]%
        {langchain}
 \bibinfo{year}{2024}\natexlab{}.
\newblock \bibinfo{title}{LangChain}.
\newblock
  \bibinfo{howpublished}{\url{https://python.langchain.com/docs/get_started/introduction}}.
    (\bibinfo{year}{2024}).
\newblock


\bibitem[\protect\citeauthoryear{??}{ope}{2024}]%
        {openai}
 \bibinfo{year}{2024}\natexlab{}.
\newblock \bibinfo{title}{OpenAI}.
\newblock \bibinfo{howpublished}{\url{https://openai.com/}}.
  (\bibinfo{year}{2024}).
\newblock


\bibitem[\protect\citeauthoryear{??}{tex}{2024}]%
        {text-embedding-3}
 \bibinfo{year}{2024}\natexlab{}.
\newblock \bibinfo{title}{OpenAI text-embedding-3 model}.
\newblock
  \bibinfo{howpublished}{\url{https://openai.com/blog/new-embedding-models-and-api-updates/}}.
    (\bibinfo{year}{2024}).
\newblock


\bibitem[\protect\citeauthoryear{??}{pin}{2024}]%
        {pinecone-faiss}
 \bibinfo{year}{2024}\natexlab{}.
\newblock \bibinfo{title}{{Pinecone: Introduction to Facebook AI Similarity
  Search (Faiss)}}.
\newblock   (\bibinfo{year}{2024}).
\newblock
\newblock
\shownote{\url{https://www.pinecone.io/learn/series/faiss/faiss-tutorial/}.}


\bibitem[\protect\citeauthoryear{??}{wik}{2024}]%
        {wikipedia_embeddings}
 \bibinfo{year}{2024}\natexlab{}.
\newblock \bibinfo{title}{Wikipedia (en) embedded with cohere.ai
  multilingual-22-12 encoder}.
\newblock
  \bibinfo{howpublished}{\url{https://huggingface.co/datasets/Cohere/wikipedia-22-12-en-embeddings/}}.
    (\bibinfo{year}{2024}).
\newblock


\bibitem[\protect\citeauthoryear{Ainslie, Lee-Thorp, de~Jong, Zemlyanskiy,
  Lebr{\'o}n, and Sanghai}{Ainslie et~al\mbox{.}}{2023}]%
        {ainslie2023gqa}
\bibfield{author}{\bibinfo{person}{Joshua Ainslie}, \bibinfo{person}{James
  Lee-Thorp}, \bibinfo{person}{Michiel de Jong}, \bibinfo{person}{Yury
  Zemlyanskiy}, \bibinfo{person}{Federico Lebr{\'o}n}, {and}
  \bibinfo{person}{Sumit Sanghai}.} \bibinfo{year}{2023}\natexlab{}.
\newblock \showarticletitle{Gqa: Training generalized multi-query transformer
  models from multi-head checkpoints}.
\newblock \bibinfo{journal}{{\em arXiv preprint arXiv:2305.13245\/}}
  (\bibinfo{year}{2023}).
\newblock


\bibitem[\protect\citeauthoryear{Babenko and Lempitsky}{Babenko and
  Lempitsky}{2014}]%
        {babenko2014inverted}
\bibfield{author}{\bibinfo{person}{Artem Babenko} {and} \bibinfo{person}{Victor
  Lempitsky}.} \bibinfo{year}{2014}\natexlab{}.
\newblock \showarticletitle{The inverted multi-index}.
\newblock \bibinfo{journal}{{\em IEEE transactions on pattern analysis and
  machine intelligence\/}} (\bibinfo{year}{2014}).
\newblock


\bibitem[\protect\citeauthoryear{Borgeaud, Mensch, Hoffmann, Cai, Rutherford,
  Millican, Van Den~Driessche, Lespiau, Damoc, Clark, et~al\mbox{.}}{Borgeaud
  et~al\mbox{.}}{2022}]%
        {borgeaud2022improving}
\bibfield{author}{\bibinfo{person}{Sebastian Borgeaud}, \bibinfo{person}{Arthur
  Mensch}, \bibinfo{person}{Jordan Hoffmann}, \bibinfo{person}{Trevor Cai},
  \bibinfo{person}{Eliza Rutherford}, \bibinfo{person}{Katie Millican},
  \bibinfo{person}{George~Bm Van Den~Driessche}, \bibinfo{person}{Jean-Baptiste
  Lespiau}, \bibinfo{person}{Bogdan Damoc}, \bibinfo{person}{Aidan Clark},
  {et~al\mbox{.}}} \bibinfo{year}{2022}\natexlab{}.
\newblock \showarticletitle{Improving language models by retrieving from
  trillions of tokens}. In \bibinfo{booktitle}{{\em International Conference on
  Machine Learning (ICML)}}.
\newblock


\bibitem[\protect\citeauthoryear{Chen, Xu, Arora, and Choi}{Chen
  et~al\mbox{.}}{2023}]%
        {chen2023understanding}
\bibfield{author}{\bibinfo{person}{Hung-Ting Chen}, \bibinfo{person}{Fangyuan
  Xu}, \bibinfo{person}{Shane~A Arora}, {and} \bibinfo{person}{Eunsol Choi}.}
  \bibinfo{year}{2023}\natexlab{}.
\newblock \showarticletitle{Understanding retrieval augmentation for long-form
  question answering}.
\newblock \bibinfo{journal}{{\em arXiv preprint arXiv:2310.12150\/}}
  (\bibinfo{year}{2023}).
\newblock


\bibitem[\protect\citeauthoryear{Chen, Lin, Han, and Sun}{Chen
  et~al\mbox{.}}{2024}]%
        {chen2024benchmarking}
\bibfield{author}{\bibinfo{person}{Jiawei Chen}, \bibinfo{person}{Hongyu Lin},
  \bibinfo{person}{Xianpei Han}, {and} \bibinfo{person}{Le Sun}.}
  \bibinfo{year}{2024}\natexlab{}.
\newblock \showarticletitle{Benchmarking large language models in
  retrieval-augmented generation}. In \bibinfo{booktitle}{{\em AAAI Conference
  on Artificial Intelligence}}.
\newblock


\bibitem[\protect\citeauthoryear{Chen, Zhao, Wang, Li, Liu, Li, Yang, and
  Wang}{Chen et~al\mbox{.}}{2021}]%
        {chen2021spann}
\bibfield{author}{\bibinfo{person}{Qi Chen}, \bibinfo{person}{Bing Zhao},
  \bibinfo{person}{Haidong Wang}, \bibinfo{person}{Mingqin Li},
  \bibinfo{person}{Chuanjie Liu}, \bibinfo{person}{Zengzhong Li},
  \bibinfo{person}{Mao Yang}, {and} \bibinfo{person}{Jingdong Wang}.}
  \bibinfo{year}{2021}\natexlab{}.
\newblock \showarticletitle{Spann: Highly-efficient billion-scale approximate
  nearest neighborhood search}.
\newblock \bibinfo{journal}{{\em Advances in Neural Information Processing
  Systems\/}} (\bibinfo{year}{2021}).
\newblock


\bibitem[\protect\citeauthoryear{Cherkasova}{Cherkasova}{1998}]%
        {cherkasova1998improving}
\bibfield{author}{\bibinfo{person}{Ludmila Cherkasova}.}
  \bibinfo{year}{1998}\natexlab{}.
\newblock \bibinfo{booktitle}{{\em Improving WWW proxies performance with
  greedy-dual-size-frequency caching policy}}.
\newblock \bibinfo{publisher}{Hewlett-Packard Laboratories Palo Alto, CA, USA}.
\newblock


\bibitem[\protect\citeauthoryear{Chowdhery, Narang, Devlin, Bosma, Mishra,
  Roberts, Barham, Chung, Sutton, Gehrmann, et~al\mbox{.}}{Chowdhery
  et~al\mbox{.}}{2022}]%
        {palm}
\bibfield{author}{\bibinfo{person}{Aakanksha Chowdhery},
  \bibinfo{person}{Sharan Narang}, \bibinfo{person}{Jacob Devlin},
  \bibinfo{person}{Maarten Bosma}, \bibinfo{person}{Gaurav Mishra},
  \bibinfo{person}{Adam Roberts}, \bibinfo{person}{Paul Barham},
  \bibinfo{person}{Hyung~Won Chung}, \bibinfo{person}{Charles Sutton},
  \bibinfo{person}{Sebastian Gehrmann}, {et~al\mbox{.}}}
  \bibinfo{year}{2022}\natexlab{}.
\newblock \showarticletitle{Palm: Scaling language modeling with pathways}.
\newblock \bibinfo{journal}{{\em arXiv preprint arXiv:2204.02311\/}}
  (\bibinfo{year}{2022}).
\newblock


\bibitem[\protect\citeauthoryear{Dettmers, Lewis, Belkada, and
  Zettlemoyer}{Dettmers et~al\mbox{.}}{2022}]%
        {dettmers2022gpt3}
\bibfield{author}{\bibinfo{person}{Tim Dettmers}, \bibinfo{person}{Mike Lewis},
  \bibinfo{person}{Younes Belkada}, {and} \bibinfo{person}{Luke Zettlemoyer}.}
  \bibinfo{year}{2022}\natexlab{}.
\newblock \showarticletitle{Gpt3. int8 (): 8-bit matrix multiplication for
  transformers at scale}.
\newblock \bibinfo{journal}{{\em Advances in Neural Information Processing
  Systems\/}} (\bibinfo{year}{2022}).
\newblock


\bibitem[\protect\citeauthoryear{Fu, Xiang, Wang, and Cai}{Fu
  et~al\mbox{.}}{2019}]%
        {fu2019fast}
\bibfield{author}{\bibinfo{person}{Cong Fu}, \bibinfo{person}{Chao Xiang},
  \bibinfo{person}{Changxu Wang}, {and} \bibinfo{person}{Deng Cai}.}
  \bibinfo{year}{2019}\natexlab{}.
\newblock \showarticletitle{Fast approximate nearest neighbor search with the
  navigating spreading-out graph}. In \bibinfo{booktitle}{{\em Proceedings of
  the VLDB Endowment}}.
\newblock


\bibitem[\protect\citeauthoryear{Ge, Zhang, Liu, Zhang, Han, and Gao}{Ge
  et~al\mbox{.}}{2023}]%
        {ge2023model}
\bibfield{author}{\bibinfo{person}{Suyu Ge}, \bibinfo{person}{Yunan Zhang},
  \bibinfo{person}{Liyuan Liu}, \bibinfo{person}{Minjia Zhang},
  \bibinfo{person}{Jiawei Han}, {and} \bibinfo{person}{Jianfeng Gao}.}
  \bibinfo{year}{2023}\natexlab{}.
\newblock \showarticletitle{Model tells you what to discard: Adaptive kv cache
  compression for llms}.
\newblock \bibinfo{journal}{{\em arXiv preprint arXiv:2310.01801\/}}
  (\bibinfo{year}{2023}).
\newblock


\bibitem[\protect\citeauthoryear{Gim, Chen, Lee, Sarda, Khandelwal, and
  Zhong}{Gim et~al\mbox{.}}{2023}]%
        {gim2023prompt}
\bibfield{author}{\bibinfo{person}{In Gim}, \bibinfo{person}{Guojun Chen},
  \bibinfo{person}{Seung-seob Lee}, \bibinfo{person}{Nikhil Sarda},
  \bibinfo{person}{Anurag Khandelwal}, {and} \bibinfo{person}{Lin Zhong}.}
  \bibinfo{year}{2023}\natexlab{}.
\newblock \showarticletitle{Prompt cache: Modular attention reuse for
  low-latency inference}.
\newblock \bibinfo{journal}{{\em arXiv preprint arXiv:2311.04934\/}}
  (\bibinfo{year}{2023}).
\newblock


\bibitem[\protect\citeauthoryear{Hendrycks, Burns, Basart, Zou, Mazeika, Song,
  and Steinhardt}{Hendrycks et~al\mbox{.}}{2020}]%
        {hendrycks2020measuring}
\bibfield{author}{\bibinfo{person}{Dan Hendrycks}, \bibinfo{person}{Collin
  Burns}, \bibinfo{person}{Steven Basart}, \bibinfo{person}{Andy Zou},
  \bibinfo{person}{Mantas Mazeika}, \bibinfo{person}{Dawn Song}, {and}
  \bibinfo{person}{Jacob Steinhardt}.} \bibinfo{year}{2020}\natexlab{}.
\newblock \showarticletitle{Measuring massive multitask language
  understanding}.
\newblock \bibinfo{journal}{{\em arXiv preprint arXiv:2009.03300\/}}
  (\bibinfo{year}{2020}).
\newblock


\bibitem[\protect\citeauthoryear{Jayaram~Subramanya, Devvrit, Simhadri,
  Krishnawamy, and Kadekodi}{Jayaram~Subramanya et~al\mbox{.}}{2019}]%
        {jayaram2019diskann}
\bibfield{author}{\bibinfo{person}{Suhas Jayaram~Subramanya},
  \bibinfo{person}{Fnu Devvrit}, \bibinfo{person}{Harsha~Vardhan Simhadri},
  \bibinfo{person}{Ravishankar Krishnawamy}, {and} \bibinfo{person}{Rohan
  Kadekodi}.} \bibinfo{year}{2019}\natexlab{}.
\newblock \showarticletitle{Diskann: Fast accurate billion-point nearest
  neighbor search on a single node}.
\newblock \bibinfo{journal}{{\em Advances in Neural Information Processing
  Systems\/}} (\bibinfo{year}{2019}).
\newblock


\bibitem[\protect\citeauthoryear{Jiang, Sablayrolles, Mensch, Bamford, Chaplot,
  Casas, Bressand, Lengyel, Lample, Saulnier, et~al\mbox{.}}{Jiang
  et~al\mbox{.}}{2023}]%
        {jiang2023mistral}
\bibfield{author}{\bibinfo{person}{Albert~Q Jiang}, \bibinfo{person}{Alexandre
  Sablayrolles}, \bibinfo{person}{Arthur Mensch}, \bibinfo{person}{Chris
  Bamford}, \bibinfo{person}{Devendra~Singh Chaplot}, \bibinfo{person}{Diego
  de~las Casas}, \bibinfo{person}{Florian Bressand}, \bibinfo{person}{Gianna
  Lengyel}, \bibinfo{person}{Guillaume Lample}, \bibinfo{person}{Lucile
  Saulnier}, {et~al\mbox{.}}} \bibinfo{year}{2023}\natexlab{}.
\newblock \showarticletitle{Mistral 7B}.
\newblock \bibinfo{journal}{{\em arXiv preprint arXiv:2310.06825\/}}
  (\bibinfo{year}{2023}).
\newblock


\bibitem[\protect\citeauthoryear{Jiang, Sablayrolles, Roux, Mensch, Savary,
  Bamford, Chaplot, Casas, Hanna, Bressand, et~al\mbox{.}}{Jiang
  et~al\mbox{.}}{2024a}]%
        {jiang2024mixtral}
\bibfield{author}{\bibinfo{person}{Albert~Q Jiang}, \bibinfo{person}{Alexandre
  Sablayrolles}, \bibinfo{person}{Antoine Roux}, \bibinfo{person}{Arthur
  Mensch}, \bibinfo{person}{Blanche Savary}, \bibinfo{person}{Chris Bamford},
  \bibinfo{person}{Devendra~Singh Chaplot}, \bibinfo{person}{Diego de~las
  Casas}, \bibinfo{person}{Emma~Bou Hanna}, \bibinfo{person}{Florian Bressand},
  {et~al\mbox{.}}} \bibinfo{year}{2024}\natexlab{a}.
\newblock \showarticletitle{Mixtral of experts}.
\newblock \bibinfo{journal}{{\em arXiv preprint arXiv:2401.04088\/}}
  (\bibinfo{year}{2024}).
\newblock


\bibitem[\protect\citeauthoryear{Jiang, Zhang, Han, Wang, Wang, and
  Kraska}{Jiang et~al\mbox{.}}{2024b}]%
        {jiang2024piperag}
\bibfield{author}{\bibinfo{person}{Wenqi Jiang}, \bibinfo{person}{Shuai Zhang},
  \bibinfo{person}{Boran Han}, \bibinfo{person}{Jie Wang},
  \bibinfo{person}{Bernie Wang}, {and} \bibinfo{person}{Tim Kraska}.}
  \bibinfo{year}{2024}\natexlab{b}.
\newblock \showarticletitle{Piperag: Fast retrieval-augmented generation via
  algorithm-system co-design}.
\newblock \bibinfo{journal}{{\em arXiv preprint arXiv:2403.05676\/}}
  (\bibinfo{year}{2024}).
\newblock


\bibitem[\protect\citeauthoryear{Joshi, Choi, Weld, and Zettlemoyer}{Joshi
  et~al\mbox{.}}{2017}]%
        {joshi2017triviaqa}
\bibfield{author}{\bibinfo{person}{Mandar Joshi}, \bibinfo{person}{Eunsol
  Choi}, \bibinfo{person}{Daniel~S Weld}, {and} \bibinfo{person}{Luke
  Zettlemoyer}.} \bibinfo{year}{2017}\natexlab{}.
\newblock \showarticletitle{Triviaqa: A large scale distantly supervised
  challenge dataset for reading comprehension}. In \bibinfo{booktitle}{{\em
  Proceedings of the 55th Annual Meeting of the Association for Computational
  Linguistics (Volume 1: Long Papers)}}.
\newblock


\bibitem[\protect\citeauthoryear{Khattab, Santhanam, Li, Hall, Liang, Potts,
  and Zaharia}{Khattab et~al\mbox{.}}{2022}]%
        {khattab2022demonstrate}
\bibfield{author}{\bibinfo{person}{Omar Khattab}, \bibinfo{person}{Keshav
  Santhanam}, \bibinfo{person}{Xiang~Lisa Li}, \bibinfo{person}{David Hall},
  \bibinfo{person}{Percy Liang}, \bibinfo{person}{Christopher Potts}, {and}
  \bibinfo{person}{Matei Zaharia}.} \bibinfo{year}{2022}\natexlab{}.
\newblock \showarticletitle{Demonstrate-search-predict: Composing retrieval and
  language models for knowledge-intensive nlp}.
\newblock \bibinfo{journal}{{\em arXiv preprint arXiv:2212.14024\/}}
  (\bibinfo{year}{2022}).
\newblock


\bibitem[\protect\citeauthoryear{Kwiatkowski, Palomaki, Redfield, Collins,
  Parikh, Alberti, Epstein, Polosukhin, Devlin, Lee, et~al\mbox{.}}{Kwiatkowski
  et~al\mbox{.}}{2019}]%
        {kwiatkowski2019natural}
\bibfield{author}{\bibinfo{person}{Tom Kwiatkowski},
  \bibinfo{person}{Jennimaria Palomaki}, \bibinfo{person}{Olivia Redfield},
  \bibinfo{person}{Michael Collins}, \bibinfo{person}{Ankur Parikh},
  \bibinfo{person}{Chris Alberti}, \bibinfo{person}{Danielle Epstein},
  \bibinfo{person}{Illia Polosukhin}, \bibinfo{person}{Jacob Devlin},
  \bibinfo{person}{Kenton Lee}, {et~al\mbox{.}}}
  \bibinfo{year}{2019}\natexlab{}.
\newblock \showarticletitle{Natural questions: a benchmark for question
  answering research}.
\newblock \bibinfo{journal}{{\em Transactions of the Association for
  Computational Linguistics\/}} (\bibinfo{year}{2019}).
\newblock


\bibitem[\protect\citeauthoryear{Kwon, Li, Zhuang, Sheng, Zheng, Yu, Gonzalez,
  Zhang, and Stoica}{Kwon et~al\mbox{.}}{2023}]%
        {vllm}
\bibfield{author}{\bibinfo{person}{Woosuk Kwon}, \bibinfo{person}{Zhuohan Li},
  \bibinfo{person}{Siyuan Zhuang}, \bibinfo{person}{Ying Sheng},
  \bibinfo{person}{Lianmin Zheng}, \bibinfo{person}{Cody~Hao Yu},
  \bibinfo{person}{Joseph Gonzalez}, \bibinfo{person}{Hao Zhang}, {and}
  \bibinfo{person}{Ion Stoica}.} \bibinfo{year}{2023}\natexlab{}.
\newblock \showarticletitle{Efficient memory management for large language
  model serving with pagedattention}. In \bibinfo{booktitle}{{\em ACM SOSP}}.
\newblock


\bibitem[\protect\citeauthoryear{Lewis, Perez, Piktus, Petroni, Karpukhin,
  Goyal, K{\"u}ttler, Lewis, Yih, Rockt{\"a}schel, et~al\mbox{.}}{Lewis
  et~al\mbox{.}}{2020}]%
        {lewis2020retrieval}
\bibfield{author}{\bibinfo{person}{Patrick Lewis}, \bibinfo{person}{Ethan
  Perez}, \bibinfo{person}{Aleksandra Piktus}, \bibinfo{person}{Fabio Petroni},
  \bibinfo{person}{Vladimir Karpukhin}, \bibinfo{person}{Naman Goyal},
  \bibinfo{person}{Heinrich K{\"u}ttler}, \bibinfo{person}{Mike Lewis},
  \bibinfo{person}{Wen-tau Yih}, \bibinfo{person}{Tim Rockt{\"a}schel},
  {et~al\mbox{.}}} \bibinfo{year}{2020}\natexlab{}.
\newblock \showarticletitle{Retrieval-augmented generation for
  knowledge-intensive nlp tasks}.
\newblock \bibinfo{journal}{{\em Advances in Neural Information Processing
  Systems\/}} (\bibinfo{year}{2020}).
\newblock


\bibitem[\protect\citeauthoryear{Li, Zhang, Andersen, and He}{Li
  et~al\mbox{.}}{2020b}]%
        {li2020improving}
\bibfield{author}{\bibinfo{person}{Conglong Li}, \bibinfo{person}{Minjia
  Zhang}, \bibinfo{person}{David~G Andersen}, {and} \bibinfo{person}{Yuxiong
  He}.} \bibinfo{year}{2020}\natexlab{b}.
\newblock \showarticletitle{Improving approximate nearest neighbor search
  through learned adaptive early termination}. In \bibinfo{booktitle}{{\em ACM
  SIGMOD}}.
\newblock


\bibitem[\protect\citeauthoryear{Li, Wallace, Shen, Lin, Keutzer, Klein, and
  Gonzalez}{Li et~al\mbox{.}}{2020a}]%
        {li2020train}
\bibfield{author}{\bibinfo{person}{Zhuohan Li}, \bibinfo{person}{Eric Wallace},
  \bibinfo{person}{Sheng Shen}, \bibinfo{person}{Kevin Lin},
  \bibinfo{person}{Kurt Keutzer}, \bibinfo{person}{Dan Klein}, {and}
  \bibinfo{person}{Joey Gonzalez}.} \bibinfo{year}{2020}\natexlab{a}.
\newblock \showarticletitle{Train big, then compress: Rethinking model size for
  efficient training and inference of transformers}. In
  \bibinfo{booktitle}{{\em International Conference on Machine Learning
  (ICML)}}.
\newblock


\bibitem[\protect\citeauthoryear{Liu, Lin, Hewitt, Paranjape, Bevilacqua,
  Petroni, and Liang}{Liu et~al\mbox{.}}{2024}]%
        {liu2024lost}
\bibfield{author}{\bibinfo{person}{Nelson~F Liu}, \bibinfo{person}{Kevin Lin},
  \bibinfo{person}{John Hewitt}, \bibinfo{person}{Ashwin Paranjape},
  \bibinfo{person}{Michele Bevilacqua}, \bibinfo{person}{Fabio Petroni}, {and}
  \bibinfo{person}{Percy Liang}.} \bibinfo{year}{2024}\natexlab{}.
\newblock \showarticletitle{Lost in the middle: How language models use long
  contexts}.
\newblock \bibinfo{journal}{{\em Transactions of the Association for
  Computational Linguistics\/}} (\bibinfo{year}{2024}).
\newblock


\bibitem[\protect\citeauthoryear{Liu, Li, Du, Yao, Cheng, Huang, Lu, Maire,
  Hoffmann, Holtzman, et~al\mbox{.}}{Liu et~al\mbox{.}}{2023}]%
        {liu2023cachegen}
\bibfield{author}{\bibinfo{person}{Yuhan Liu}, \bibinfo{person}{Hanchen Li},
  \bibinfo{person}{Kuntai Du}, \bibinfo{person}{Jiayi Yao},
  \bibinfo{person}{Yihua Cheng}, \bibinfo{person}{Yuyang Huang},
  \bibinfo{person}{Shan Lu}, \bibinfo{person}{Michael Maire},
  \bibinfo{person}{Henry Hoffmann}, \bibinfo{person}{Ari Holtzman},
  {et~al\mbox{.}}} \bibinfo{year}{2023}\natexlab{}.
\newblock \showarticletitle{CacheGen: Fast Context Loading for Language Model
  Applications}.
\newblock \bibinfo{journal}{{\em arXiv preprint arXiv:2310.07240\/}}
  (\bibinfo{year}{2023}).
\newblock


\bibitem[\protect\citeauthoryear{Liu, Desai, Liao, Wang, Xie, Xu, Kyrillidis,
  and Shrivastava}{Liu et~al\mbox{.}}{2024}]%
        {liu2024scissorhands}
\bibfield{author}{\bibinfo{person}{Zichang Liu}, \bibinfo{person}{Aditya
  Desai}, \bibinfo{person}{Fangshuo Liao}, \bibinfo{person}{Weitao Wang},
  \bibinfo{person}{Victor Xie}, \bibinfo{person}{Zhaozhuo Xu},
  \bibinfo{person}{Anastasios Kyrillidis}, {and} \bibinfo{person}{Anshumali
  Shrivastava}.} \bibinfo{year}{2024}\natexlab{}.
\newblock \showarticletitle{Scissorhands: Exploiting the persistence of
  importance hypothesis for llm kv cache compression at test time}.
\newblock \bibinfo{journal}{{\em Advances in Neural Information Processing
  Systems\/}} (\bibinfo{year}{2024}).
\newblock


\bibitem[\protect\citeauthoryear{Lu, Duan, Han, Guo, Hwang, and
  Svyatkovskiy}{Lu et~al\mbox{.}}{2022}]%
        {lu2022reacc}
\bibfield{author}{\bibinfo{person}{Shuai Lu}, \bibinfo{person}{Nan Duan},
  \bibinfo{person}{Hojae Han}, \bibinfo{person}{Daya Guo},
  \bibinfo{person}{Seung-won Hwang}, {and} \bibinfo{person}{Alexey
  Svyatkovskiy}.} \bibinfo{year}{2022}\natexlab{}.
\newblock \showarticletitle{Reacc: A retrieval-augmented code completion
  framework}.
\newblock \bibinfo{journal}{{\em arXiv preprint arXiv:2203.07722\/}}
  (\bibinfo{year}{2022}).
\newblock


\bibitem[\protect\citeauthoryear{Malkov and Yashunin}{Malkov and
  Yashunin}{2018}]%
        {malkov2018efficient}
\bibfield{author}{\bibinfo{person}{Yu~A Malkov} {and} \bibinfo{person}{Dmitry~A
  Yashunin}.} \bibinfo{year}{2018}\natexlab{}.
\newblock \showarticletitle{Efficient and robust approximate nearest neighbor
  search using hierarchical navigable small world graphs}.
\newblock \bibinfo{journal}{{\em IEEE transactions on pattern analysis and
  machine intelligence\/}} (\bibinfo{year}{2018}).
\newblock


\bibitem[\protect\citeauthoryear{OpenAI}{OpenAI}{2023}]%
        {gpt4}
\bibfield{author}{\bibinfo{person}{OpenAI}.} \bibinfo{year}{2023}\natexlab{}.
\newblock \showarticletitle{GPT-4 Technical Report}.
\newblock \bibinfo{journal}{{\em arXiv preprint arXiv:2303.08774\/}}
  (\bibinfo{year}{2023}).
\newblock


\bibitem[\protect\citeauthoryear{Paszke, Gross, Massa, Lerer, Bradbury, Chanan,
  Killeen, Lin, Gimelshein, Antiga, et~al\mbox{.}}{Paszke
  et~al\mbox{.}}{2019}]%
        {paszke2019pytorch}
\bibfield{author}{\bibinfo{person}{Adam Paszke}, \bibinfo{person}{Sam Gross},
  \bibinfo{person}{Francisco Massa}, \bibinfo{person}{Adam Lerer},
  \bibinfo{person}{James Bradbury}, \bibinfo{person}{Gregory Chanan},
  \bibinfo{person}{Trevor Killeen}, \bibinfo{person}{Zeming Lin},
  \bibinfo{person}{Natalia Gimelshein}, \bibinfo{person}{Luca Antiga},
  {et~al\mbox{.}}} \bibinfo{year}{2019}\natexlab{}.
\newblock \showarticletitle{Pytorch: An imperative style, high-performance deep
  learning library}.
\newblock \bibinfo{journal}{{\em Advances in neural information processing
  systems\/}} (\bibinfo{year}{2019}).
\newblock


\bibitem[\protect\citeauthoryear{Ram, Levine, Dalmedigos, Muhlgay, Shashua,
  Leyton-Brown, and Shoham}{Ram et~al\mbox{.}}{2023}]%
        {ram2023context}
\bibfield{author}{\bibinfo{person}{Ori Ram}, \bibinfo{person}{Yoav Levine},
  \bibinfo{person}{Itay Dalmedigos}, \bibinfo{person}{Dor Muhlgay},
  \bibinfo{person}{Amnon Shashua}, \bibinfo{person}{Kevin Leyton-Brown}, {and}
  \bibinfo{person}{Yoav Shoham}.} \bibinfo{year}{2023}\natexlab{}.
\newblock \showarticletitle{In-context retrieval-augmented language models}.
\newblock \bibinfo{journal}{{\em Transactions of the Association for
  Computational Linguistics\/}} (\bibinfo{year}{2023}).
\newblock


\bibitem[\protect\citeauthoryear{Shortle, Thompson, Gross, and Harris}{Shortle
  et~al\mbox{.}}{2018}]%
        {shortle2018fundamentals}
\bibfield{author}{\bibinfo{person}{John~F Shortle}, \bibinfo{person}{James~M
  Thompson}, \bibinfo{person}{Donald Gross}, {and} \bibinfo{person}{Carl~M
  Harris}.} \bibinfo{year}{2018}\natexlab{}.
\newblock \bibinfo{booktitle}{{\em Fundamentals of queueing theory}}.
  Vol.~\bibinfo{volume}{399}.
\newblock \bibinfo{publisher}{John Wiley \& Sons}.
\newblock


\bibitem[\protect\citeauthoryear{Siriwardhana, Weerasekera, Wen, Kaluarachchi,
  Rana, and Nanayakkara}{Siriwardhana et~al\mbox{.}}{2023}]%
        {siriwardhana2023improving}
\bibfield{author}{\bibinfo{person}{Shamane Siriwardhana},
  \bibinfo{person}{Rivindu Weerasekera}, \bibinfo{person}{Elliott Wen},
  \bibinfo{person}{Tharindu Kaluarachchi}, \bibinfo{person}{Rajib Rana}, {and}
  \bibinfo{person}{Suranga Nanayakkara}.} \bibinfo{year}{2023}\natexlab{}.
\newblock \showarticletitle{Improving the domain adaptation of retrieval
  augmented generation (RAG) models for open domain question answering}.
\newblock \bibinfo{journal}{{\em Transactions of the Association for
  Computational Linguistics\/}} (\bibinfo{year}{2023}).
\newblock


\bibitem[\protect\citeauthoryear{Tillet, Kung, and Cox}{Tillet
  et~al\mbox{.}}{2019}]%
        {tillet2019triton}
\bibfield{author}{\bibinfo{person}{Philippe Tillet},
  \bibinfo{person}{Hsiang-Tsung Kung}, {and} \bibinfo{person}{David Cox}.}
  \bibinfo{year}{2019}\natexlab{}.
\newblock \showarticletitle{Triton: an intermediate language and compiler for
  tiled neural network computations}. In \bibinfo{booktitle}{{\em Proceedings
  of the 3rd ACM SIGPLAN International Workshop on Machine Learning and
  Programming Languages}}.
\newblock


\bibitem[\protect\citeauthoryear{Touvron, Lavril, Izacard, Martinet, Lachaux,
  Lacroix, Rozi{\`e}re, Goyal, Hambro, Azhar, et~al\mbox{.}}{Touvron
  et~al\mbox{.}}{2023}]%
        {touvron2023llama}
\bibfield{author}{\bibinfo{person}{Hugo Touvron}, \bibinfo{person}{Thibaut
  Lavril}, \bibinfo{person}{Gautier Izacard}, \bibinfo{person}{Xavier
  Martinet}, \bibinfo{person}{Marie-Anne Lachaux},
  \bibinfo{person}{Timoth{\'e}e Lacroix}, \bibinfo{person}{Baptiste
  Rozi{\`e}re}, \bibinfo{person}{Naman Goyal}, \bibinfo{person}{Eric Hambro},
  \bibinfo{person}{Faisal Azhar}, {et~al\mbox{.}}}
  \bibinfo{year}{2023}\natexlab{}.
\newblock \showarticletitle{Llama: Open and efficient foundation language
  models. CoRR, abs/2302.13971, 2023. doi: 10.48550}.
\newblock \bibinfo{journal}{{\em arXiv preprint arXiv.2302.13971\/}}
  (\bibinfo{year}{2023}).
\newblock


\bibitem[\protect\citeauthoryear{Trivedi, Balasubramanian, Khot, and
  Sabharwal}{Trivedi et~al\mbox{.}}{2022}]%
        {trivedi2022interleaving}
\bibfield{author}{\bibinfo{person}{Harsh Trivedi}, \bibinfo{person}{Niranjan
  Balasubramanian}, \bibinfo{person}{Tushar Khot}, {and}
  \bibinfo{person}{Ashish Sabharwal}.} \bibinfo{year}{2022}\natexlab{}.
\newblock \showarticletitle{Interleaving retrieval with chain-of-thought
  reasoning for knowledge-intensive multi-step questions}.
\newblock \bibinfo{journal}{{\em arXiv preprint arXiv:2212.10509\/}}
  (\bibinfo{year}{2022}).
\newblock


\bibitem[\protect\citeauthoryear{Vaithilingam, Zhang, and
  Glassman}{Vaithilingam et~al\mbox{.}}{2022}]%
        {vaithilingam2022expectation}
\bibfield{author}{\bibinfo{person}{Priyan Vaithilingam},
  \bibinfo{person}{Tianyi Zhang}, {and} \bibinfo{person}{Elena~L Glassman}.}
  \bibinfo{year}{2022}\natexlab{}.
\newblock \showarticletitle{Expectation vs. experience: Evaluating the
  usability of code generation tools powered by large language models}. In
  \bibinfo{booktitle}{{\em Chi conference on human factors in computing systems
  extended abstracts}}.
\newblock


\bibitem[\protect\citeauthoryear{Vaswani, Shazeer, Parmar, Uszkoreit, Jones,
  Gomez, Kaiser, and Polosukhin}{Vaswani et~al\mbox{.}}{2017}]%
        {vaswani2017attention}
\bibfield{author}{\bibinfo{person}{Ashish Vaswani}, \bibinfo{person}{Noam
  Shazeer}, \bibinfo{person}{Niki Parmar}, \bibinfo{person}{Jakob Uszkoreit},
  \bibinfo{person}{Llion Jones}, \bibinfo{person}{Aidan~N Gomez},
  \bibinfo{person}{{\L}ukasz Kaiser}, {and} \bibinfo{person}{Illia
  Polosukhin}.} \bibinfo{year}{2017}\natexlab{}.
\newblock \showarticletitle{Attention is all you need}.
\newblock \bibinfo{journal}{{\em Advances in neural information processing
  systems\/}} (\bibinfo{year}{2017}).
\newblock


\bibitem[\protect\citeauthoryear{Wu, Zhong, Zhang, Huang, Liu, and Jin}{Wu
  et~al\mbox{.}}{2023}]%
        {wu2023fast}
\bibfield{author}{\bibinfo{person}{Bingyang Wu}, \bibinfo{person}{Yinmin
  Zhong}, \bibinfo{person}{Zili Zhang}, \bibinfo{person}{Gang Huang},
  \bibinfo{person}{Xuanzhe Liu}, {and} \bibinfo{person}{Xin Jin}.}
  \bibinfo{year}{2023}\natexlab{}.
\newblock \showarticletitle{Fast distributed inference serving for large
  language models}.
\newblock \bibinfo{journal}{{\em arXiv preprint arXiv:2305.05920\/}}
  (\bibinfo{year}{2023}).
\newblock


\bibitem[\protect\citeauthoryear{Xiao, Tian, Chen, Han, and Lewis}{Xiao
  et~al\mbox{.}}{2023}]%
        {xiao2023efficient}
\bibfield{author}{\bibinfo{person}{Guangxuan Xiao}, \bibinfo{person}{Yuandong
  Tian}, \bibinfo{person}{Beidi Chen}, \bibinfo{person}{Song Han}, {and}
  \bibinfo{person}{Mike Lewis}.} \bibinfo{year}{2023}\natexlab{}.
\newblock \showarticletitle{Efficient streaming language models with attention
  sinks}.
\newblock \bibinfo{journal}{{\em arXiv preprint arXiv:2309.17453\/}}
  (\bibinfo{year}{2023}).
\newblock


\bibitem[\protect\citeauthoryear{Yang, Qi, Zhang, Bengio, Cohen, Salakhutdinov,
  and Manning}{Yang et~al\mbox{.}}{2018}]%
        {yang2018hotpotqa}
\bibfield{author}{\bibinfo{person}{Zhilin Yang}, \bibinfo{person}{Peng Qi},
  \bibinfo{person}{Saizheng Zhang}, \bibinfo{person}{Yoshua Bengio},
  \bibinfo{person}{William~W. Cohen}, \bibinfo{person}{Ruslan Salakhutdinov},
  {and} \bibinfo{person}{Christopher~D. Manning}.}
  \bibinfo{year}{2018}\natexlab{}.
\newblock \showarticletitle{{HotpotQA}: A Dataset for Diverse, Explainable
  Multi-hop Question Answering}. In \bibinfo{booktitle}{{\em Conference on
  Empirical Methods in Natural Language Processing ({EMNLP})}}.
\newblock


\bibitem[\protect\citeauthoryear{Yao, Yu, Zhao, Shafran, Griffiths, Cao, and
  Narasimhan}{Yao et~al\mbox{.}}{2024}]%
        {yao2024tree}
\bibfield{author}{\bibinfo{person}{Shunyu Yao}, \bibinfo{person}{Dian Yu},
  \bibinfo{person}{Jeffrey Zhao}, \bibinfo{person}{Izhak Shafran},
  \bibinfo{person}{Tom Griffiths}, \bibinfo{person}{Yuan Cao}, {and}
  \bibinfo{person}{Karthik Narasimhan}.} \bibinfo{year}{2024}\natexlab{}.
\newblock \showarticletitle{Tree of thoughts: Deliberate problem solving with
  large language models}.
\newblock \bibinfo{journal}{{\em Advances in Neural Information Processing
  Systems\/}} (\bibinfo{year}{2024}).
\newblock


\bibitem[\protect\citeauthoryear{Ye, Tao, Huang, and Li}{Ye
  et~al\mbox{.}}{2024}]%
        {ye2024chunkattention}
\bibfield{author}{\bibinfo{person}{Lu Ye}, \bibinfo{person}{Ze Tao},
  \bibinfo{person}{Yong Huang}, {and} \bibinfo{person}{Yang Li}.}
  \bibinfo{year}{2024}\natexlab{}.
\newblock \showarticletitle{ChunkAttention: Efficient Self-Attention with
  Prefix-Aware KV Cache and Two-Phase Partition}.
\newblock \bibinfo{journal}{{\em arXiv preprint arXiv:2402.15220\/}}
  (\bibinfo{year}{2024}).
\newblock


\bibitem[\protect\citeauthoryear{Yu, Jeong, Kim, Kim, and Chun}{Yu
  et~al\mbox{.}}{2022}]%
        {yu2022orca}
\bibfield{author}{\bibinfo{person}{Gyeong-In Yu}, \bibinfo{person}{Joo~Seong
  Jeong}, \bibinfo{person}{Geon-Woo Kim}, \bibinfo{person}{Soojeong Kim}, {and}
  \bibinfo{person}{Byung-Gon Chun}.} \bibinfo{year}{2022}\natexlab{}.
\newblock \showarticletitle{Orca: A Distributed Serving System for
  $\{$Transformer-Based$\}$ Generative Models}. In \bibinfo{booktitle}{{\em
  USENIX OSDI}}.
\newblock


\bibitem[\protect\citeauthoryear{Zhang, Haddow, and Birch}{Zhang
  et~al\mbox{.}}{2023a}]%
        {zhang2023prompting}
\bibfield{author}{\bibinfo{person}{Biao Zhang}, \bibinfo{person}{Barry Haddow},
  {and} \bibinfo{person}{Alexandra Birch}.} \bibinfo{year}{2023}\natexlab{a}.
\newblock \showarticletitle{Prompting large language model for machine
  translation: A case study}. In \bibinfo{booktitle}{{\em International
  Conference on Machine Learning (ICML)}}.
\newblock


\bibitem[\protect\citeauthoryear{Zhang, Ladhak, Durmus, Liang, McKeown, and
  Hashimoto}{Zhang et~al\mbox{.}}{2024a}]%
        {zhang2024benchmarking}
\bibfield{author}{\bibinfo{person}{Tianyi Zhang}, \bibinfo{person}{Faisal
  Ladhak}, \bibinfo{person}{Esin Durmus}, \bibinfo{person}{Percy Liang},
  \bibinfo{person}{Kathleen McKeown}, {and} \bibinfo{person}{Tatsunori~B
  Hashimoto}.} \bibinfo{year}{2024}\natexlab{a}.
\newblock \showarticletitle{Benchmarking large language models for news
  summarization}.
\newblock \bibinfo{journal}{{\em Transactions of the Association for
  Computational Linguistics\/}} (\bibinfo{year}{2024}).
\newblock


\bibitem[\protect\citeauthoryear{Zhang, Jin, Tang, Liu, and Jin}{Zhang
  et~al\mbox{.}}{2023b}]%
        {zhang2023fast}
\bibfield{author}{\bibinfo{person}{Zili Zhang}, \bibinfo{person}{Chao Jin},
  \bibinfo{person}{Linpeng Tang}, \bibinfo{person}{Xuanzhe Liu}, {and}
  \bibinfo{person}{Xin Jin}.} \bibinfo{year}{2023}\natexlab{b}.
\newblock \showarticletitle{Fast, Approximate Vector Queries on Very Large
  Unstructured Datasets}. In \bibinfo{booktitle}{{\em USENIX NSDI}}.
\newblock


\bibitem[\protect\citeauthoryear{Zhang, Liu, Huang, Liu, and Jin}{Zhang
  et~al\mbox{.}}{2024b}]%
        {zhang2024fast}
\bibfield{author}{\bibinfo{person}{Zili Zhang}, \bibinfo{person}{Fangyue Liu},
  \bibinfo{person}{Gang Huang}, \bibinfo{person}{Xuanzhe Liu}, {and}
  \bibinfo{person}{Xin Jin}.} \bibinfo{year}{2024}\natexlab{b}.
\newblock \showarticletitle{Fast Vector Query Processing for Large Datasets
  Beyond GPU Memory with Reordered Pipelining}. In \bibinfo{booktitle}{{\em
  USENIX NSDI}}.
\newblock


\bibitem[\protect\citeauthoryear{Zhang, Sheng, Zhou, Chen, Zheng, Cai, Song,
  Tian, R{\'e}, Barrett, et~al\mbox{.}}{Zhang et~al\mbox{.}}{2024c}]%
        {zhang2024h2o}
\bibfield{author}{\bibinfo{person}{Zhenyu Zhang}, \bibinfo{person}{Ying Sheng},
  \bibinfo{person}{Tianyi Zhou}, \bibinfo{person}{Tianlong Chen},
  \bibinfo{person}{Lianmin Zheng}, \bibinfo{person}{Ruisi Cai},
  \bibinfo{person}{Zhao Song}, \bibinfo{person}{Yuandong Tian},
  \bibinfo{person}{Christopher R{\'e}}, \bibinfo{person}{Clark Barrett},
  {et~al\mbox{.}}} \bibinfo{year}{2024}\natexlab{c}.
\newblock \showarticletitle{H2o: Heavy-hitter oracle for efficient generative
  inference of large language models}.
\newblock \bibinfo{journal}{{\em Advances in Neural Information Processing
  Systems\/}} (\bibinfo{year}{2024}).
\newblock


\bibitem[\protect\citeauthoryear{Zhang, Zhu, Yang, Xu, Li, Phothilimthana, and
  Jia}{Zhang et~al\mbox{.}}{2024d}]%
        {zhang2024accelerating}
\bibfield{author}{\bibinfo{person}{Zhihao Zhang}, \bibinfo{person}{Alan Zhu},
  \bibinfo{person}{Lijie Yang}, \bibinfo{person}{Yihua Xu},
  \bibinfo{person}{Lanting Li}, \bibinfo{person}{Phitchaya~Mangpo
  Phothilimthana}, {and} \bibinfo{person}{Zhihao Jia}.}
  \bibinfo{year}{2024}\natexlab{d}.
\newblock \showarticletitle{Accelerating retrieval-augmented language model
  serving with speculation}.
\newblock \bibinfo{journal}{{\em arXiv preprint arXiv:2401.14021\/}}
  (\bibinfo{year}{2024}).
\newblock


\bibitem[\protect\citeauthoryear{Zheng, Yin, Xie, Huang, Sun, Yu, Cao,
  Kozyrakis, Stoica, Gonzalez, et~al\mbox{.}}{Zheng et~al\mbox{.}}{2023}]%
        {sglang}
\bibfield{author}{\bibinfo{person}{Lianmin Zheng}, \bibinfo{person}{Liangsheng
  Yin}, \bibinfo{person}{Zhiqiang Xie}, \bibinfo{person}{Jeff Huang},
  \bibinfo{person}{Chuyue Sun}, \bibinfo{person}{Cody~Hao Yu},
  \bibinfo{person}{Shiyi Cao}, \bibinfo{person}{Christos Kozyrakis},
  \bibinfo{person}{Ion Stoica}, \bibinfo{person}{Joseph~E Gonzalez},
  {et~al\mbox{.}}} \bibinfo{year}{2023}\natexlab{}.
\newblock \showarticletitle{Efficiently programming large language models using
  sglang}.
\newblock \bibinfo{journal}{{\em arXiv preprint arXiv:2312.07104\/}}
  (\bibinfo{year}{2023}).
\newblock


\bibitem[\protect\citeauthoryear{Zhong, Liu, Chen, Hu, Zhu, Liu, Jin, and
  Zhang}{Zhong et~al\mbox{.}}{2024}]%
        {zhong2024distserve}
\bibfield{author}{\bibinfo{person}{Yinmin Zhong}, \bibinfo{person}{Shengyu
  Liu}, \bibinfo{person}{Junda Chen}, \bibinfo{person}{Jianbo Hu},
  \bibinfo{person}{Yibo Zhu}, \bibinfo{person}{Xuanzhe Liu},
  \bibinfo{person}{Xin Jin}, {and} \bibinfo{person}{Hao Zhang}.}
  \bibinfo{year}{2024}\natexlab{}.
\newblock \showarticletitle{DistServe: Disaggregating Prefill and Decoding for
  Goodput-optimized Large Language Model Serving}.
\newblock \bibinfo{journal}{{\em arXiv preprint arXiv:2401.09670\/}}
  (\bibinfo{year}{2024}).
\newblock


\end{thebibliography}

\end{document}